\documentclass[journal,draftcls,onecolumn,12pt,twoside]{IEEEtran}
\usepackage{cite}
\usepackage{amsmath}
\usepackage{amssymb}
\usepackage{amsthm} 
\usepackage{epstopdf}
\usepackage{booktabs}
\usepackage[table,xcdraw]{xcolor}
\usepackage{graphicx}
\usepackage{float}
\usepackage{color}
\usepackage{latexsym}
\usepackage{algorithm}

\usepackage[nodisplayskipstretch]{setspace}
\usepackage{optidef}
\usepackage[noend]{algpseudocode}
\renewcommand\IEEEkeywordsname{Keywords}

\begin{document}

\setstretch{1.49}

\title{Optimizing Downlink Resource Allocation in Multiuser MIMO Networks via Fractional Programming and the Hungarian Algorithm}

\author{Ahmad~Ali~Khan,~\IEEEmembership{Student Member,~IEEE,}
        Raviraj~Adve,~\IEEEmembership{Fellow,~IEEE,}
        and~Wei~Yu,~\IEEEmembership{Fellow,~IEEE}
        
\thanks{The authors are with the Department
of Electrical and Computer Engineering, University of Toronto, Ontario,
ON M5S 3G4, Canada. E-mails: (akhan, rsadve, weiyu)@ece.utoronto.ca. The materials in this paper have been presented in part at IEEE Global Communications Conference (Globecom) $7^\text{th}$ International Workshop on "Emerging Technologies for 5G and Beyond Wireless and Mobile Networks (ET5GB)", Abu Dhabi, December 2018 \cite{khan2018optimizing}.}
}

\markboth{IEEE Transactions on Wireless Communications}%
{Submitted paper}
\maketitle

\vspace*{-4em}

%
\IEEEpeerreviewmaketitle

{
\begin{abstract}
        Optimizing the sum-log-utility for the downlink of multi-frequency band, multiuser, multiantenna networks requires joint solutions to the associated beamforming and user scheduling problems through the use of cloud radio access network (CRAN) architecture; optimizing such a network is, however, non-convex and NP-hard. In this paper, we present a novel iterative beamforming and scheduling strategy based on fractional programming and the Hungarian algorithm. The beamforming strategy allows us to iteratively maximize the chosen objective function in a fashion similar to block coordinate ascent. Furthermore, based on the crucial insight that, in the downlink, the interference pattern remains fixed for a given set of beamforming weights, we use the Hungarian algorithm as an efficient approach to optimally schedule users for the given set of beamforming weights. Specifically, this approach allows us to select the best subset of users (amongst the larger set of all available users). Our simulation results show that, in terms of average sum-log-utility, as well as sum-rate, the proposed scheme substantially outperforms both the state-of-the-art multicell weighted minimum mean-squared error (WMMSE) and greedy proportionally fair WMMSE schemes, as well as standard interior-point and sequential quadratic solvers. Importantly, our proposed scheme is also far more computationally efficient than the multicell WMMSE scheme.
\end{abstract}}

\section{Introduction}
\baselineskip 7.5mm

The improvements in spectral efficiency, throughput and quality-of-service achieved by utilizing multiantenna networks have been extensively documented in the literature~\cite{bolcskei2006mimo},~\cite{hoydis2013massive}. In particular, optimizing the resource allocation in such multiantenna networks, by designing beamforming weights and scheduling specific users from the larger pool of potential users, is central to fully exploiting the finite wireless resources available and maximizing spectral efficiency~\cite{bjornson2013optimal},~\cite{dahrouj2010coordinated}. However, designing efficient resource allocation schemes remains challenging, since many utility functions of practical interest, such as sum-rate, sum-log-utility and min-rate, are inherently nonconvex functions of transmit powers; in fact, the associated optimization problems for each of these objective functions have been found to be NP-hard~\cite{luo2008dynamic}. Thus, solving these optimization problems to global optimality entails impractical computational complexity even for very small network sizes~\cite{liu2012achieving}.

One solution to these problems is to utilize single-cell schemes, such as zero-forcing or matched filtering, in which base stations ignore intercell interference when designing beamforming weights and making scheduling decisions \cite{suh2011downlink}. While far from globally optimal, such uncoordinated schemes offer three significant advantages: first, they are analytically tractable in the sense that they can be analyzed using tools from probability theory and stochastic geometry to yield accurate estimates of the data rates achieved by users (and hence objective functions like the network sum-rate)~\cite{hosseini2018optimizing}; second, these schemes are computationally efficient, especially compared to globally optimal techniques or iterative block coordinate descent based algorithms like weighted minimum mean squared error (WMMSE)~\cite{shi2011iteratively} processing; third, and likely most important, in these schemes each base-station (BS) requires channel state information (CSI) from only its own users, not for users in other cells\footnote{We note that, like other works that focus on algorithm development~\cite{luo2008dynamic, shi2011iteratively}, the acquistion of CSI, its overhead and quality is beyond the scope of this paper. However, we do acknowledge that this is a vitally important problem in wireless networks.}. As such, these uncoordinated schemes offer a useful benchmark against which to evaluate the performance of more sophisticated resource allocation strategies. 

Coordinated resource allocation schemes, in which base stations jointly design their scheduling and beamforming decisions, improve on uncoordinated schemes. {Such joint design leads to improved quality-of-service since it helps to mitigate the effects of both inter-cell and intra-cell interference \cite{an2017achieving}}. In doing so, however, such schemes inevitably incur increased computational complexity, as compared to uncoordinated schemes, since these algorithms optimize across multiple BSs. Since the objective functions for most utility maximization problems are nonconvex, such schemes typically rely on block coordinate ascent~\cite{shi2011iteratively}, successive convex approximation~\cite{weeraddana2012weighted} or other heuristic methods \cite{park2018iterative},~\cite{douik2016coordinated} to reach, at best, a local optimum. 

A number of coordinated schemes have been developed in the literature; for example, in~\cite{yu2013multicell} the authors develop an interference pricing and greedy proportionally fair (PF) scheduling algorithm to maximize the weighted sum rate (WSR) for the downlink. The proposed scheme demonstrates excellent performance in terms of average sum-log-utility but is not guaranteed to be nondecreasing in the objective function since greedy scheduling is used. In~\cite{weeraddana2012weighted}, Weeraddana et al.~propose an algorithm based on the successive convex approximation approach to optimize the needed beamforming weights and power allocation in order to solve the general WSR maximization problem for the downlink of a multiple input multiple output (MIMO) cellular network. The algorithm requires minimal exchange of information between cooperating BSs; however, the algorithm is also shown to underperform the WMMSE scheme of~\cite{shi2011iteratively}. Additionally, one alternative is to employ worst-case weighted sum-rate maximization \cite{chinnadurai2018worst}, although such an approach is generally better suited to settings with uncertainty in channel vectors.

The work in~\cite{shi2011iteratively} develops the WMMSE algorithm by demonstrating the equivalence of minimizing the weighted MSE and maximizing the WSR and adopting a block coordinate descent strategy to reach a (guaranteed) local optimum of the original WSR objective function. The algorithm iterates between obtaining beamforming weights and a set of auxiliary variables, optimizing one while the other is kept fixed. This WMMSE approach demonstrates excellent performance and is, thus, widely utilized as a benchmark against which the performance of other coordinated resource allocation schemes is compared. 

The work in~\cite{shi2011iteratively} does not address the important problem of user scheduling, i.e., choosing a set of users to serve from the larger set of available users. One solution is to use the multicell WMMSE scheme, with \emph{all users across all cells} are scheduled $-$ the BSs then jointly design beamforming weights for each and every user, as in \cite{shi2011iteratively}. Eventually, after a number of iterations, the power assigned (the norm of the beamforming vectors) to most users will be essentially zero and these users are, then, \emph{implicitly} not scheduled. However, since this set of ``unscheduled" users is unknown a priori, beamforming weights have to be calculated for \emph{all users} in the network for \textit{each iteration of the algorithm}. This is extremely computationally expensive since a large matrix needs to be inverted in each step. Furthermore, it is worth emphasizing that the multicell WMMSE scheme, as described, is \emph{not globally optimal}. Additionally, as the authors in \cite{shen2018fractional} have observed, when scheduling all users, the WMMSE algorithm tends to get stuck in a low-quality locally optimal solution. Despite these drawbacks, the WMMSE algorithm remains the benchmark against which other resource allocation algorithms are evaluated \cite{shen2018fractional2,kaleva2016decentralized,li2015new}.

A lower-complexity approach is to alternately optimize the scheduling and beamforming variables in a fashion similar to that proposed by~\cite{douik2018joint} and~\cite{zhang2017sum}, by alternately utilizing the WMMSE algorithm for beamforming and updating the scheduling decisions using the greedy PF scheme. However, because of the greedy step, this approach is also not guaranteed to be nondecreasing in the original WSR objective function. 

Globally optimal schemes to solve sum-rate and weighted-sum-rate optimization problems have also been formulated in the literature, using the framework of monotonic optimization~\cite{liu2012achieving, bjornson2013optimal, brehmer2012utility, utschick2012monotonic}, as well as geometric and arithmetic-mean methods \cite{roshandeh2018exact}. For example, the authors in~\cite{liu2012achieving} develop an algorithm to find the globally optimal beamformers to maximize the WSR for the downlink of a multiuser multiantenna network. However, as the authors in~\cite{brehmer2012utility} note, these globally optimal schemes require impractical computational complexity for even small systems, and are thus used almost exclusively as benchmarks for very small network sizes. 

{It is worth emphasizing that extensive CSI exchange between BSs and computational resources are required in order to enable both locally and globally optimal resource management schemes. These requirements are best served by the utilization of cloud radio access networks (CRANs), which allow for flexible deployment of resource allocation algorithms and on-demand processing while utilizing relatively inexpensive hardware at the BS \cite{quek2017cloud,peng2016recent,peng2014energy,qian2015baseband,gerasimenko2015cooperative}. Deploying dedicated hardware at the BS level to implement individual algorithms is both technically challenging and cost ineffective \cite{quek2017cloud,gerasimenko2015cooperative}; on the other hand, through the use of CRAN, low-cost remote radio heads can be utilized at each BS, while virtualized baseband processing units for the entire network can be implemented in the cloud and easily altered to enable different resource management schemes and capabilities \cite{qian2015baseband,peng2014energy}. Thus, the CRAN architecture is necessary in order to enable coordinated resource allocation and is stated explicitly or assumed implictly in the various coordinated schemes detailed in the literature \cite{weeraddana2012weighted}.}

In summary, effective multicell resource allocation schemes with relatively low complexity are, as yet, not available. It is this gap in the literature that we address here. Specifically, we develop an iterative scheduling and beamforming strategy to find an effective solution to the problem of maximizing the average sum-log-utility function for the downlink of a multiuser MIMO network. Using the framework of fractional programming, originally developed in~\cite{shen2018fractional2} and~\cite{shen2018fractional} for uplink problems and extended to the matrix setting in~\cite{shen2018coordinated}, we derive a scheme similar to a block coordinate ascent scheme. 
In \cite{shen2018fractional2}, fractional programming has shown large performance benefits for utility maximization in the \textit{uplink} setting. {Similarly, in \cite{zhang2018energy}, the authors utilize fractional programming to jointly optimize power control and scheduling decisions for energy efficiency maximization; the proposed algorithm can be implemented in both distributed and centralized fashion and provides excellent convergence and performance properties. In both these scenarios, the interference pattern changes with the scheduling decisions; thus, optimization across multiple cells can provide considerable benefit. This paper demonstrates the efficacy of fractional programming for the \textit{downlink}, where we exploit the \textit{fixed} interference pattern to improve performance and reduce computational complexity in the user scheduling step. 

In contrast to our conference-length work in ~\cite{khan2018optimizing}, this paper considers the most general setting of the problem: we derive the algorithm and present results for proportionally-fair WSR and sum-rate maximization through scheduling and beamforming across multiple frequency bands with both joint and decoupled power constraints across the bands. Deriving the algorithm for this setting is considerably more challenging than the single-band case considered in \cite{khan2018optimizing}; this is especially true for the joint power allocation across multiple bands in which, despite the orthogonality of the bands, \textit{all} beamforming weights and scheduling decisions become coupled due to the power constraint. Nonetheless, we demonstrate that fractional programming allows us to decouple these optimization variables and solve for an effective solution with guaranteed nondecreasing convergence. Specifically, the contributions of this paper are:}
\begin{itemize}
\item We formulate the downlink sum-log-utility maximization problem as a WSR problem in the general case of multiple interfering cells, multiple frequency bands, and a large number of potential users per cell.
\item We develop a joint beamforming and user scheduling algorithm based on fractional programming and the Hungarian algorithm. The Hungarian algorithm selects the optimal set of users from the much larger pool of potential users, for a given set of beamforming weights, in \textit{polynomial} time. The development of these two aspects of our overall algorithm is our key contribution, as the scheduling step allows us to reach an effective solution while simultaneously reducing computational complexity.
\item We compare the performance of joint power allocation across all frequency bands with the simpler case in which power constraints are de-coupled across bands. We show that the simpler approach, in fact, suffers little performance loss.
\item We show that each iteration of the proposed algorithm leads to nondecreasing objective function values; the overall algorithm outperforms several competing approaches, {including the state-of-the-art multicell WMMSE, as well as standard interior-point and sequential quadratic programming solvers widely utilized in the literature, with significantly lower computational complexity}.
\item Our proposed algorithm outperforms the aforementioned competing state-of-the-art approaches over a wide range of BS maximum transmit power values.
\end{itemize}

This paper is organized as follows: In Section II, we present our system model and formulate the desired optimization problem. In Section III, we describe the proposed solution approach in detail, while also presenting a proof for its convergence. In Section IV, we present the results and compare the performance and computational complexity of the proposed scheme against the benchmarks described previously. We draw some conclusions in Section V.

{Prior to proceeding further, we define some notation used in this paper. $\mathbb{R}$, ${\mathbb{R}}_{{+}}$ and ${\mathbb{R}}_{{++}}$ represent the set of real numbers, non-negative real numbers and positive numbers respectively. We denote scalars using lowercase (eg. $x$), vectors using lowercase boldface (eg. $\mathbf{x}$), matrices using uppercase boldface (eg. $\mathbf{X}$) and sets using script typeface (eg. $\mathcal{X}$). The operator $\left|{.}\right|$ denotes the absolute value when applied to a scalar and cardinality when applied to a set; we use ${\left\|{.}\right\|}_{2}$ to denote the ${\ell}_{2}$-norm of a vector. The conjugate of a complex scalar $z$ is denoted by $z^{*}$; the Hermitian of a complex vector $\mathbf{z}$ is denoted by $\mathbf{z}^{H}$. Likewise, $\mathbb{C}$ represents the set of complex numbers. A complex multivariate normal distribution with mean $\boldsymbol{\mu}$ and covariance matrix $\mathbf{K}$ is denoted by ${\mathcal{CN}}\left({\boldsymbol{\mu}{,}\mathbf{K}}\right)$. Finally, $\mathbf{I}$ represents the identity matrix.}

\section{System Model and Problem Formulation}
{We consider the downlink of a wireless cellular network, with base-stations located in a regular hexagonal pattern; we denote the set of BSs in the network by $\mathcal{B}$. Each user associates with the geographically closest BS, with ${K}_{b}$ users associating with the ${b}^{\mathrm{th}}$ base-station; under the hexagonal grid layout of the BSs, this leads to identically sized hexagonal cells. We choose this hexagonal pattern purely for convenience; the derivations and algorithms that follow are applicable to any distribution of BSs in a multi-cell network}. We assume that {there are a total of $F$} orthogonal frequency bands, each of bandwidth $W$, available for transmission to each base-station. Each BS is equipped with $M$  transmit antennas which are capable of simultaneously transmitting on all available frequency bands; each user is equipped with a single receive antenna capable of simultaneously receiving signals on all available frequency bands. Furthermore, we also assume that the number of users associated with each BS significantly exceeds the number of transmit antennas available at the base-station (i.e., ${K}_{b}\mathrm{\gg}{M}$ for all $b$). Figure \ref{LS_MIMO_opt_problem_figure} illustrates the system at hand with hexagonal cells.

{
\begin{figure}[!]
        \begin{center} 
                \setlength\belowcaptionskip{-3.7\baselineskip}
                \includegraphics[trim={0 0cm 0 0cm},clip,height=0.44\textwidth]{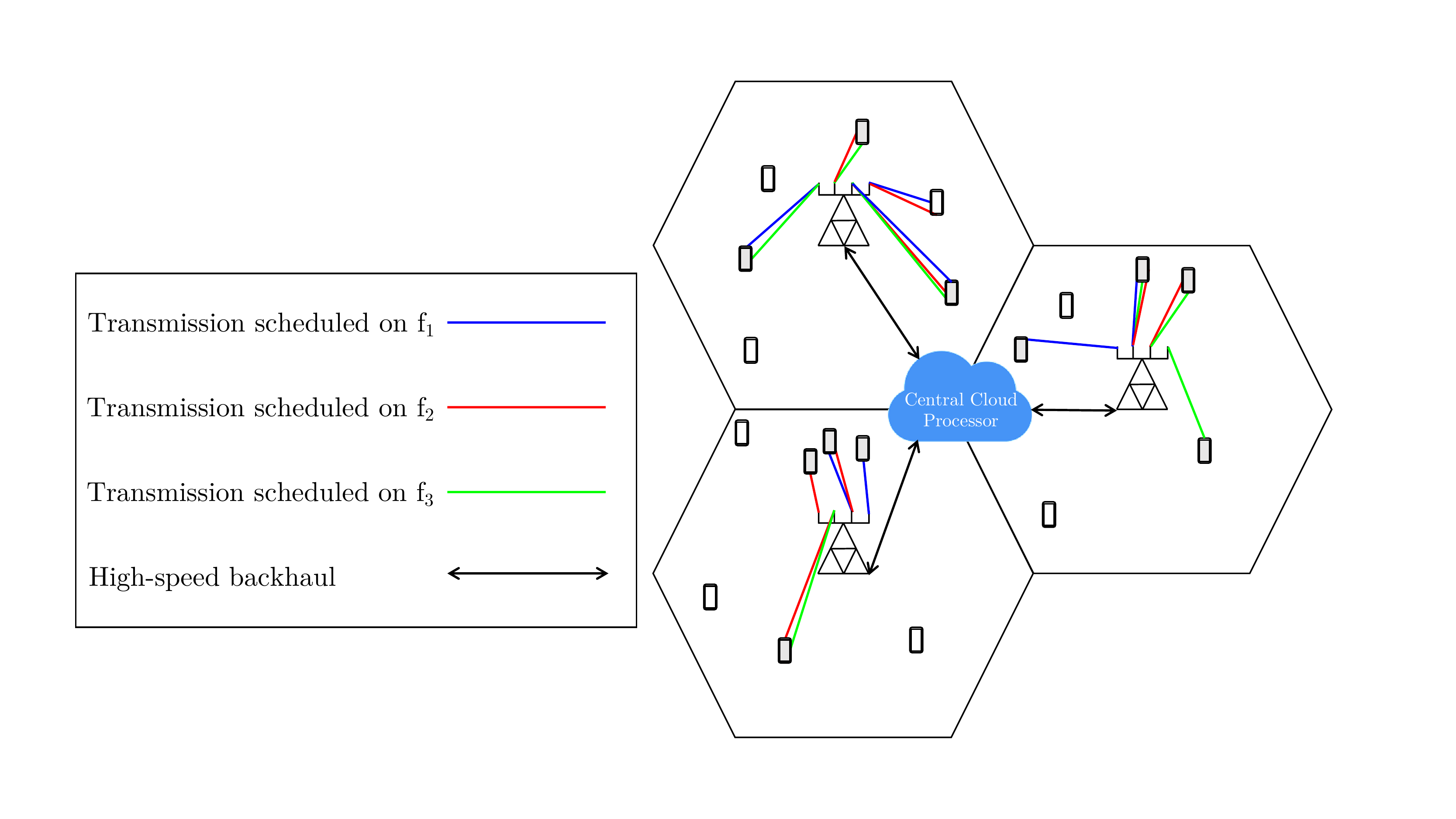}
                \centering
                \caption{{Network model for the proposed system model. As shown, each BS in the network serves $K$ users within its cell with $M$ transmit antennas, and is thus capable of serving multiple users on each of $F=3$ orthogonal frequency bands. BSs are connected to a central cloud processor via high-speed backhaul and forward their downlink CSI. This processor computes the scheduling and beamforming strategy and forwards it back to the BSs.}}
                \centering
                \centering
                \label{LS_MIMO_opt_problem_figure}
        \end{center}
\end{figure} 
}

Prior to stating the channel model, it is important to emphasize that this paper focuses on analyzing the best system-level performance. In this regard, we make two important assumptions that are common to the papers in this area~\cite{yu2013multicell},\cite{shen2018coordinated}. First, we assume that all the BSs have access to perfect CSI of all their associated users on all frequency bands. Second, we assume that all the BSs are connected via high-speed backhaul links to a central cloud server that is capable of performing system-level optimization based on the CSI received from each of the BSs, and relaying back the beamforming weight vectors and scheduling information; {thus our system model falls within the general realm of CRAN}. This is necessary for a coordinated transmission strategy to achieve the best system-level performance. 
{
The downlink channel from the ${b'}^{\mathrm{th}}$ base-station to the ${k}^{\mathrm{th}}$  user associated with the ${b}^{\mathrm{th}}$ BS on the ${f}^{\mathrm{th}}$ frequency band is a complex ${M}\mathrm{\times}{1}$ vector denoted by ${\mathbf{h}}_{kb\mathrm{,}b'\mathrm{,}f}$. As stated earlier, each BS serves only the set of users associated with it. Accordingly, the beamforming weight vector for the ${k}^{\mathrm{th}}$ user associated with the ${b}^{\mathrm{th}}$ BS on the ${f}^{\mathrm{th}}$ frequency band is a complex ${M}\mathrm{\times}{1}$ vector denoted as ${\mathbf{v}}_{kb\mathrm{,}f}$. 

In each time slot, each base-station schedules a \emph{subset} of its associated users. Specifically, we impose the constraint that each base-station can schedule no more than $M$ users per time slot on each available frequency band. The binary variable ${u}_{kb\mathrm{,}f}$ is used to indicate whether the ${k}^{\mathrm{th}}$ user associated with the ${b}^{\mathrm{th}}$ BS is scheduled ($u_{kb,f}=1$) or not ($u_{kb,f}=0$) on the ${f}^{\mathrm{th}}$ frequency band. The symbol intended for the ${k}^{\mathrm{th}}$ user associated with the ${b}^{\mathrm{th}}$ BS on the ${f}^{\mathrm{th}}$ frequency band is a complex scalar denoted by ${s}_{kb\mathrm{,}f}$. It follows that the received downlink signal ${r}_{kb\mathrm{,}f}$ at the  ${k}^{\mathrm{th}}$ user associated with the ${b}^{\mathrm{th}}$ base-station on the ${f}^{\mathrm{th}}$ frequency band is given by}
{
\begin{equation}
{r}_{kb,f}{=}{\mathbf{h}}_{kb,b,f}^{H}{\mathbf{v}}_{kb,f}{u}_{kb,f}{s}_{kb,f}{+}\mathop{\sum}\limits_{\mathop{{b}{'}{=}{1}}\limits_{{(}{b}\prime{,}{k}\prime{)}\ne}}\limits^{\left|{\mathcal{B}}\right|\hspace{0.33em}}{\mathop{\sum}\limits_{\mathop{{k}{'}{=}{1}}\limits_{(b,k)}}\limits^{{K}_{b'}}{{\mathbf{h}}_{kb,b',f}^{H}{\mathbf{v}}_{k'b',f}{u}_{k'b',f}{s}_{k'b',f}}}{+}{z}_{kb,f}
\end{equation}
where ${z}_{kb,f}$ denotes the additive zero-mean Gaussian noise with variance ${\mathit{\sigma}}_{kb\mathrm{,}f}^{2}$.
Thus, the signal-to-interference-plus-noise ratio for the user under consideration on the ${f}^{th}$ frequency band is given by
\[
{\mathit{\gamma}}_{kb\mathrm{,}f}\mathrm{{{}={}}}\frac{{u}_{kb\mathrm{,}f}\mathrm{|}{\mathbf{h}}_{kb\mathrm{,}b\mathrm{,}f}^{H}{\mathbf{v}}_{kb\mathrm{,}f}{\mathrm{|}}^{2}}{\mathop{\sum}\limits_{\mathop{{b}{\mathrm{'}}\mathrm{{{}={}}}{1}}\limits_{{\mathrm{(}}{b}\prime{\mathrm{,}}{k}\prime{\mathrm{)}}\kern0.2em\mathrm{\ne}}}\limits^{\left|{\mathcal{B}}\right|\hspace{0.33em}}{\mathop{\sum}\limits_{\mathop{{k}{\mathrm{'}}\mathrm{{{}={}}}{1}}\limits_{\mathrm{(}b\mathrm{,}k\mathrm{)}}}\limits^{{K}_{b\mathrm{'}}}{{u}_{k\mathrm{'}b\mathrm{',}f}{\mathrm{|}}{\mathbf{h}}_{kb\mathrm{,}b\mathrm{',}f}^{H}{\mathbf{v}}_{k\mathrm{'}b\mathrm{',}f}{\mathrm{|}}^{2}\mathrm{{+}}{\mathit{\sigma}}_{kb\mathrm{,}f}^{2}}}}
\]
}
Consequently, the data rate to the user on the ${f}^{th}$ frequency band is given by ${R}_{kb\mathrm{,}f}\mathrm{{{}={}}}{W}\log{\mathrm{(}}{1}\mathrm{{+}}{\mathit{\gamma}}_{kb\mathrm{,}f}{\mathrm{)}}$. Recalling that we have a total of $F$ frequency bands available to serve the user, it is clear that the combined data rate to the user in a time slot, denoted by ${R}_{kb\mathrm{,}total}$, is given by ${R}_{kb\mathrm{,}total}\mathrm{{{}={}}}\mathop{\sum}\limits_{{f}\mathrm{{{}={}}}{1}}\limits^{F}{{R}_{kb\mathrm{,}f}}$.

Now, we formulate the resource allocation problem. Our choice of the objective function is the network WSR.

The maximization of the weighted sum rate with an appropriate choice of weights in each time slot can lead to approximate maximization of arbitrary network utility functions. A popular network utility function is the sum of the logarithm of the long-term average data rates achieved by each of the users. This leads to a proportionally fair allocation of resources amongst all users in the network.
\newpage
We are interested in answering the following question: given a network as described above, to maximize the WSR in each time slot which subset of users should each BS serve, in which frequency band, at what power level and with what beamformer design? The optimization problem that encapsulates this question for a single time slot can be expressed as{
\begin{subequations} \label{LS_MIMO_equiv_problem}
        \begin{align}
        \mathop{\mathrm{maximize}}\limits_{{\mathbf{U}}{\mathbf{,}}{\mathbf{V}}{\mathbf{,}}\mathbf{\Gamma}}\hspace{2.16em}\mathop{\sum}\limits_{{b}\mathrm{{{}={}}}{1}}\limits^{\left|{\mathcal{B}}\right|}{\mathop{\sum}\limits_{{k}\mathrm{{{}={}}}{1}}\limits^{{K}_{b}}{\mathop{\sum}\limits_{{f}\mathrm{{{}={}}}{1}}\limits^{F}{{w}_{kb}\log\left({{1}\mathrm{{+}}{\mathit{\gamma}}_{kb\mathrm{,}f}}\right)}}}\hspace{17.78em}\label{aux_objective_LS_MIMO}\\
        {\mathrm{subject}}\hspace{0.33em}{\mathrm{to}}\hspace{2.20em}{u}_{kb\mathrm{,}f}\mathrm{\in}{\mathrm{\{}}{0}{\mathrm{,}}{1}{\mathrm{\},}}\hspace{6.55em}{b}\mathrm{{{}={}}}{1}{\mathrm{,...,}}\hspace{0.33em}\left|{\mathcal{B}}\right|{\mathrm{;}}\hspace{0.33em}{k}\mathrm{{{}={}}}{1}{\mathrm{,...,}}\hspace{0.33em}{K}_{b}{\mathrm{;}}\hspace{0.33em}{f}\mathrm{{{}={}}}{1}{\mathrm{,...,}}\hspace{0.33em}{F}\hspace{0.47em}\label{c1_LS_MIMO_equiv}\\
        \mathop{\sum}\limits_{{k}\mathrm{{{}={}}}{1}}\limits^{{K}_{b}}{{u}_{kb\mathrm{,}f}}\hspace{0.33em}\mathrm{\leq}\hspace{0.33em}{M}{\mathrm{,}}\hspace{5.44em}{b}\mathrm{{{}={}}}{1}{\mathrm{,...,}}\hspace{0.33em}\left|{\mathcal{B}}\right|{\mathrm{;}}\hspace{0.33em}{f}\mathrm{{{}={}}}{1}{\mathrm{,...,}}\hspace{0.33em}{F}\hspace{6.50em}\label{c2_LS_MIMO_equiv}\\
        \mathop{\sum}\limits_{{f}\mathrm{{{}={}}}{1}}\limits^{F}{\mathop{\sum}\limits_{{k}\mathrm{{{}={}}}{1}}\limits^{{K}_{b}}{{\left\|{{\mathbf{v}}_{kb,f}}\right\|}_{2}^{2}\hspace{0.33em}\mathrm{\leq}\hspace{0.33em}{F}{P}_{T}}}{\mathrm{,}}\hspace{1.58em}{b}\mathrm{{{}={}}}{1}{\mathrm{,...,}}\hspace{0.33em}\left|{\mathcal{B}}\right|\hspace{12.20em}\label{c3_LS_MIMO_equiv}\\
        {\mathit{\gamma}}_{kb\mathrm{,}f}\mathrm{{{}={}}}\frac{{u}_{kb\mathrm{,}f}\mathrm{|}{\mathbf{h}}_{kb\mathrm{,}b\mathrm{,}f}^{H}{\mathbf{v}}_{kb\mathrm{,}f}{\mathrm{|}}^{2}}{\mathop{\sum}\limits_{\mathop{{b}{\mathrm{'}}\mathrm{{{}={}}}{1}}\limits_{{\mathrm{(}}{b}\prime{\mathrm{,}}{k}\prime{\mathrm{)}}\kern0.2em\mathrm{\ne}}}\limits^{\left|{\mathcal{B}}\right|\hspace{0.33em}}{\mathop{\sum}\limits_{\mathop{{k}{\mathrm{'}}\mathrm{{{}={}}}{1}}\limits_{\mathrm{(}b\mathrm{,}k\mathrm{)}}}\limits^{{K}_{b\mathrm{'}}}{{u}_{k\mathrm{'}b\mathrm{'}}{\mathrm{|}}{\mathbf{h}}_{kb\mathrm{,}b\mathrm{'}\mathrm{,}f}^{H}{\mathbf{v}}_{k\mathrm{'}b\mathrm{'}\mathrm{,}f}{\mathrm{|}}^{2}\mathrm{{+}}{\mathit{\sigma}}_{kb\mathrm{,}f}^{2}}}}\hspace{2.33em}\begin{array}{c}{{b}{{}={}}{1}{,}\ldots{,}\left|{\mathcal{B}}\right|{;}}\\{{k}{{}={}}{1}{,}\ldots{,}{K}_{b}{;}\hspace{0.50em}}\\{{f}{{}={}}{1}{,}\ldots{,}{F}\hspace{1.25em}}\end{array}\hspace{1.80em}\label{equality_constraint_aux_objective_LS_MIMO_equiv}
        \end{align}
\end{subequations}    

Here, ${w}_{kb}$ represents the weight for the ${k}^{\mathrm{th}}$ user associated with the ${b}^{\mathrm{th}}$ base-station\footnote{The proportionally fair weight for the ${n}^{\mathrm{th}}$ time slot is usually determined by finding the inverse of the long-term average data rate achieved by the user in question over an exponentially decaying window \cite{yu2010book}, i.e.,
$\label{updated_weights_LS_MIMO}{w}_{kb}^{(n)}{=}{1}{/}{\bar{R}}_{kb}^{(n)}$,
\noindent where ${\bar{R}}_{kb}^{(n)}$ represents the exponentially weighted average data rate achieved by the user in the time slots preceding the ${n}^{\mathrm{th}}$ time slot across all frequency bands. This is calculated using the update equation ${\bar{R}}_{kb}^{(n)}{{}={}}{(}{1}{-}\mathit{\alpha}{)}{\bar{R}}_{kb}^{{(}{n}{-}{1}{)}}{+}\mathit{\alpha}{R}_{kb}^{(n)}$,
where ${R}_{kb}^{(n)}$ represents the total data rate achieved in the $n^{th}$ time slot across all bands, and $\alpha$ represents the forgetting factor. 
}, while $\mathbf{U}$ and $\mathbf{V}$ denote the optimization variables gathered into a matrix: the scheduling variables ($\mathbf{U}$) and the beamformers ($\mathbf{V}$). For simplicity of notation, we drop the index denoting the time slot in the formulation of the optimization problem in (\ref{LS_MIMO_equiv_problem}). The SINR of the ${k}^{\mathrm{th}}$ user associated with the ${b}^{\mathrm{th}}$ BS on the ${f}^{\mathrm{th}}$ frequency band is denoted by ${\mathit{\gamma}}_{kb\mathrm{,}f}$.

Our objective function is the network WSR for a single time slot as expressed in (\ref{aux_objective_LS_MIMO}). The constraint in (\ref{c1_LS_MIMO_equiv}) enforces the scheduling decisions by the BS to be binary; a user is scheduled if its scheduling variable equals one, and vice versa. The second set of constraints in (\ref{c2_LS_MIMO_equiv}) ensures that a BS is restricted to serving a subset of at most \textit{M} users from the set of all associated users on each available frequency band. The constraints in (\ref{c3_LS_MIMO_equiv}) impose a total transmit power constraint ${FP}_{T}$ at each BS across the different frequency bands (and thus an average power constraint of ${P}_{T}$ across each individual frequency band). Finally, the equality constraints in (\ref{equality_constraint_aux_objective_LS_MIMO_equiv}) enforce the SINR values at each user, BS and frequency band.}
\section{Proposed Approach} \label{ss proposed algorithm single freq}
We note that the optimization problem in~\eqref{LS_MIMO_equiv_problem} has a mixed-binary integer form and is nonconvex in the beamforming variables. In fact, as stated earlier, the general WSR maximization problem has been shown to be NP-hard by Luo and Zhang in~\cite{luo2008dynamic}. {We also note that the beamforming and scheduling variables are coupled across the different frequency bands due to the sum-power constraint in (\ref{c3_LS_MIMO_equiv})), even though the bands themselves are orthogonal and thus noninterfering}. To solve this problem and obtain an effective solution, we adopt an iterative optimization strategy, based on fractional programming as developed by Shen and Yu in~\cite{shen2018fractional} and~\cite{shen2018fractional2}. 

{We begin by introducing Lagrange multipliers for each of the equality constraints in (\ref{equality_constraint_aux_objective_LS_MIMO_equiv}), in a similar fashion to \cite{shen2018fractional2}, as follows}

\begin{equation} \label{lagrangian_LS_MIMO}
{\mathcal{L}}{\mathrm{(}}{\mathbf{U}}{\mathbf{,}}{\mathbf{V}}{\mathbf{,}}\mathbf{\Gamma}{\mathbf{,}}\mathbf{\Lambda}{\mathrm{)}}\mathrm{{{}={}}}\hspace{0.33em}\mathop{\sum}\limits_{{b}\mathrm{{{}={}}}{1}}\limits^{\left|{\mathcal{B}}\right|}{\mathop{\sum}\limits_{{k}\mathrm{{{}={}}}{1}}\limits^{{K}_{b}}{\mathop{\sum}\limits_{{f}\mathrm{{{}={}}}{1}}\limits^{F}{\left[{{w}_{kb}\log\left({{1}\mathrm{{+}}{\mathit{\gamma}}_{kb\mathrm{,}f}}\right)\mathrm{{-}}{\mathit{\lambda}}_{kb\mathrm{,}f}\left({{\mathit{\gamma}}_{kb\mathrm{,}f}\mathrm{{-}}\frac{{u}_{kb\mathrm{,}f}\mathrm{|}{\mathbf{h}}_{kb\mathrm{,}b\mathrm{,}f}^{H}{\mathbf{v}}_{kb\mathrm{,}f}{\mathrm{|}}^{2}}{\mathop{\sum}\limits_{\mathop{{b}{\mathrm{'}}\mathrm{{{}={}}}{1}}\limits_{{\mathrm{(}}{b}\prime{\mathrm{,}}{k}\prime{\mathrm{)}}\kern0.2em\mathrm{\ne}}}\limits^{\left|{\mathcal{B}}\right|\hspace{0.33em}}{\mathop{\sum}\limits_{\mathop{{k}{\mathrm{'}}\mathrm{{{}={}}}{1}}\limits_{\mathrm{(}b\mathrm{,}k\mathrm{)}}}\limits^{{K}_{b\mathrm{'}}}{{u}_{k\mathrm{'}b\mathrm{',}f}{\mathrm{|}}{\mathbf{h}}_{kb\mathrm{,}b\mathrm{',}f}^{H}{\mathbf{v}}_{k\mathrm{'}b\mathrm{',}f}{\mathrm{|}}^{2}\mathrm{{+}}{\mathit{\sigma}}_{kb\mathrm{,}f}^{2}}}}}\right)}\right]}}}
\end{equation}
where the ${\mathit{\lambda}}_{kb\mathrm{,}f}$ represent the Lagrange multipliers for each of the equality constraints in (\ref{equality_constraint_aux_objective_LS_MIMO_equiv}).
For notational clarity, the SINR auxiliary variables and Lagrange multipliers are collected in the matrices $\mathbf{\Gamma}$ and $\mathbf{\Lambda}$ respectively. Consequently, in order to satisfy the first-order condition with respect to the $\gamma_{kb,f}$ values, we set the partial derivative with respect to the Lagrangian equal to zero, i.e.,
\begin{equation} \label{optimal_gamma_LS_MIMO}
\frac{\partial\mathcal{L}(\mathbf{U},\mathbf{V},\boldsymbol{\Gamma},\boldsymbol{\Lambda})}{\partial\gamma_{kb,f}} = 0.
\end{equation}
Now substituting (\ref{optimal_gamma_LS_MIMO}) into the equality constraint (\ref{equality_constraint_aux_objective_LS_MIMO_equiv}), we then obtain the optimal Lagrange multipliers as
\begin{equation} \label{optimal_lagrange_multiplier_single_frequency}
{\mathit{\lambda}}_{kb\mathrm{,}f\mathrm{,}opt}\mathrm{{{}={}}}\frac{{w}_{kb}\left({\mathop{\sum}\limits_{\mathop{{b}{\mathrm{'}}\mathrm{{{}={}}}{1}}\limits_{{\mathrm{(}}{b}\prime{\mathrm{,}}{k}\prime{\mathrm{)}}\kern0.2em\mathrm{\ne}}}\limits^{\left|{\mathcal{B}}\right|\hspace{0.33em}}{\mathop{\sum}\limits_{\mathop{{k}{\mathrm{'}}\mathrm{{{}={}}}{1}}\limits_{\mathrm{(}b\mathrm{,}k\mathrm{)}}}\limits^{{K}_{b\mathrm{'}}}{{u}_{k\mathrm{'}b\mathrm{'}\mathrm{,}f}{\mathrm{|}}{\mathbf{h}}_{kb\mathrm{,}b\mathrm{'}\mathrm{,}f}^{H}{\mathbf{v}}_{k\mathrm{'}b\mathrm{'}\mathrm{,}f}{\mathrm{|}}^{2}\mathrm{{+}}{\mathit{\sigma}}_{kb\mathrm{,}f}^{2}}}}\right)}{\mathop{\sum}\limits_{{b}{\mathrm{'}}\mathrm{{{}={}}}{1}}\limits^{\left|{\mathcal{B}}\right|}{\mathop{\sum}\limits_{{k}{\mathrm{'}}\mathrm{{{}={}}}{1}}\limits^{{K}_{b\mathrm{'}}}{{u}_{k\mathrm{'}b\mathrm{'}\mathrm{,}f}{\mathrm{|}}{\mathbf{h}}_{kb\mathrm{,}b\mathrm{'}\mathrm{,}f}^{H}{\mathbf{v}}_{k\mathrm{'}b\mathrm{'}\mathrm{,}f}{\mathrm{|}}^{2}\mathrm{{+}}{\mathit{\sigma}}_{kb\mathrm{,}f}^{2}}}}\hspace{0.33em}
\end{equation}
Substituting the optimal Lagrange multipliers from (\ref{optimal_lagrange_multiplier_single_frequency}) into the expression for the Lagrangian in (\ref{lagrangian_LS_MIMO}), we obtain the following reformulated objective function, which we denote by ${f}_{r}{\mathrm{(}}{\mathbf{U}}{\mathbf{,}}{\mathbf{V}}{\mathbf{,}}\mathbf{\Gamma}{\mathrm{)}}$.
\begin{equation} \label{f_r_single_frequency}
\begin{gathered}
{{f}_{r}{\mathrm{(}}{\mathbf{U}}{\mathrm{,}}{\mathbf{V}}{\mathrm{,}}\mathbf{\Gamma}{\mathrm{)}}\mathrm{{{}={}}}\mathop{\sum}\limits_{{b}\mathrm{{{}={}}}{1}}\limits^{\left|{\mathcal{B}}\right|}{\mathop{\sum}\limits_{{k}\mathrm{{{}={}}}{1}}\limits^{{K}_{b}}{\mathop{\sum}\limits_{{f}\mathrm{{{}={}}}{1}}\limits^{F}{{w}_{kb}\log\left({{1}\mathrm{{+}}{\mathit{\gamma}}_{kb\mathrm{,}f}}\right)\mathrm{{-}}\mathop{\sum}\limits_{{b}\mathrm{{{}={}}}{1}}\limits^{\left|{\mathcal{B}}\right|}{\mathop{\sum}\limits_{{k}\mathrm{{{}={}}}{1}}\limits^{{K}_{b}}{\mathop{\sum}\limits_{{f}\mathrm{{{}={}}}{1}}\limits^{F}{{w}_{kb}{\mathit{\gamma}}_{kb\mathrm{,}f}}}}}}}} \hspace{06.00em}\\
{\mathrm{{+}}\mathop{\sum}\limits_{{b}\mathrm{{{}={}}}{1}}\limits^{\left|{\mathcal{B}}\right|}{\mathop{\sum}\limits_{{k}\mathrm{{{}={}}}{1}}\limits^{{K}_{b}}{\mathop{\sum}\limits_{{f}\mathrm{{{}={}}}{1}}\limits^{F}{\frac{{w}_{kb}{\mathrm{(}}{1}\mathrm{{+}}{\mathit{\gamma}}_{kb\mathrm{,}f}{\mathrm{)}}{u}_{kb\mathrm{,}f}{\mathrm{|}}{\mathbf{h}}_{kb\mathrm{,}b\mathrm{,}f}^{H}{\mathbf{v}}_{kb\mathrm{,}f}{\mathrm{|}}^{2}}{\mathop{\sum}\limits_{{b}{\mathrm{'}}\mathrm{{{}={}}}{1}}\limits^{\left|{\mathcal{B}}\right|}{\mathop{\sum}\limits_{{k}{\mathrm{'}}\mathrm{{{}={}}}{1}}\limits^{{K}_{b\mathrm{'}}}{{u}_{k\mathrm{'}b\mathrm{',}f}{\mathrm{|}}{\mathbf{h}}_{kb\mathrm{,}b\mathrm{',}f}^{H}{\mathbf{v}}_{k\mathrm{'}b\mathrm{',}f}{\mathrm{|}}^{2}\mathrm{{+}}{\mathit{\sigma}}_{kb\mathrm{,}f}^{2}}}}}}}\hspace{1.90em}}
\end{gathered}
\end{equation}
Thus, it follows that the original optimization problem in (\ref{LS_MIMO_equiv_problem}) can be expressed as the following reformulated optimization problem
\begin{subequations} \label{first_reformulated_problem_LS_MIMO}
        \begin{align}
        \mathop{\mathrm{maximize}}\limits_{{\mathbf{U}}{\mathrm{,}}{\mathbf{V}}{\mathbf{,}}\mathbf{\Gamma}}\hspace{2.30em}{f}_{r}{\mathrm{(}}{\mathbf{U}}{\mathbf{,}}{\mathbf{V}}{\mathbf{,}}\mathbf{\Gamma}{\mathrm{)}}\hspace{24.80em}\\
        {\mathrm{subject}}\hspace{0.33em}{\mathrm{to}}\hspace{2.20em}{u}_{kb\mathrm{,}f}\mathrm{\in}{\mathrm{\{}}{0}{\mathrm{,}}{1}{\mathrm{\},}}\hspace{6.69em}{b}\mathrm{{{}={}}}{1}{\mathrm{,...,}}\hspace{0.33em}\left|{\mathcal{B}}\right|{\mathrm{;}}\hspace{0.33em}{k}\mathrm{{{}={}}}{1}{\mathrm{,...,}}\hspace{0.33em}{K}_{b}{\mathrm{;}}\hspace{0.33em}{f}\mathrm{{{}={}}}{1}{\mathrm{,...,}}\hspace{0.33em}{F}\hspace{0.33em}\\
        \mathop{\sum}\limits_{{k}\mathrm{{{}={}}}{1}}\limits^{{K}_{b}}{{u}_{kb\mathrm{,}f}}\hspace{0.33em}\mathrm{\leq}\hspace{0.33em}{M}{\mathrm{,}}\hspace{5.54em}{b}\mathrm{{{}={}}}{1}{\mathrm{,...,}}\hspace{0.33em}\left|{\mathcal{B}}\right|{\mathrm{;}}\hspace{0.33em}{f}\mathrm{{{}={}}}{1}{\mathrm{,...,}}\hspace{0.33em}{F}\hspace{6.30em}\\
        \mathop{\sum}\limits_{{f}\mathrm{{{}={}}}{1}}\limits^{F}{\mathop{\sum}\limits_{{k}\mathrm{{{}={}}}{1}}\limits^{{K}_{b}}{{\left\|{{\mathbf{v}}_{kb}}\right\|}_{2}^{2}\hspace{0.33em}\mathrm{\leq}\hspace{0.33em}{F}{P}_{T}}}{\mathrm{,}}\hspace{2.14em}{b}\mathrm{{{}={}}}{1}{\mathrm{,...,}}\hspace{0.33em}\left|{\mathcal{B}}\right|\hspace{12.08em}
        \end{align}
\end{subequations}
 We note that the reformulated optimization problem in (\ref{first_reformulated_problem_LS_MIMO}) is equivalent to the original optimization problem in (\ref{LS_MIMO_equiv_problem}) in the sense that the optimal objective function value and the associated primal optimization variables, $\mathbf{U}$ and $\mathbf{V}$, for both problems are identical.

A key point to note before we proceed further is that the ratio terms which were present inside the logarithm function in problem (\ref{LS_MIMO_equiv_problem}) have now been moved outside as the sum-of-ratios term in the reformulated objective function. This is the first step in allowing us to develop an iterative optimization strategy.
In order to proceed further, we make use of the following theorem, as a vector-valued version of that derived by Shen and Yu in~\cite{shen2018fractional2}:
\vspace{-1.50em}
\newtheorem{theorem}{Theorem}
\begin{theorem} \label{Thm1}
        \textit{Let} ${n}_{i}{\mathrm{(}}{\mathbf{x}}{\mathrm{)}}{\mathrm{:}}\hspace{0.33em}{\mathbb{C}}^{m}\hspace{0.1em}\mathrm{\longmapsto}\hspace{0.1em}{\mathbb{C}}$ \textit{and} ${d}_{i}{\mathrm{(}}{\mathbf{x}}{\mathrm{)}}{\mathrm{:}}\hspace{0.33em}{\mathbb{C}}^{m}\hspace{0.1em}\mathrm{\longmapsto}\hspace{0.1em}{\mathbb{R}}_{\mathrm{{+}}\mathrm{{+}}}$, \textit{where} ${i}\mathrm{{{}={}}}{1}{\mathrm{,...,}}\hspace{0.33em}{N}$ \textit{be two functions of the optimization variables $\mathbf{x}$ and} $
        {m}\hspace{0.1em}\mathrm{\in}\hspace{0.1em}{\mathbb{N}}$\textit{. Furthermore, let} ${\mathcal{X}}\mathrm{\subseteq}{\mathbb{C}}^{m}$ \textit{be a constraint set. Then the sum-of-ratios optimization problem}
        \begin{equation} \label{quad_transform_original_objective}
        \begin{gathered}
        {\mathop{\mathrm{maximize}}\limits_{\mathbf{x}}\hspace{0.33em}\hspace{0.33em}\hspace{0.60em}\mathop{\sum}\limits_{{i}\mathrm{{{}={}}}{1}}\limits^{N}{\frac{\mathrm{|}{\mathit{n}}_{\mathit{i}}{\mathrm{(}\mathbf{x}\mathrm{)}\mathrm{|}}^{2}}{{\mathit{d}}_{\mathit{i}}\mathrm{(}\mathbf{x}\mathrm{)}}}\hspace{0.33em}\hspace{0.33em}\hspace{0.33em}\hspace{0.33em}\hspace{0.33em}\hspace{0.33em}\hspace{0.33em}\hspace{0.33em}\hspace{0.33em}\hspace{0.33em}\hspace{0.33em}\hspace{0.33em}\hspace{0.33em}} \hfill\\
        {{\mathrm{subject}}\hspace{0.33em}{\mathrm{to}}\hspace{0.00em}\hspace{1.26em}{\mathbf{x}}\mathrm{\in}{\mathcal{X}}\hspace{0.33em}\hspace{0.33em}\hspace{0.33em}\hspace{0.33em}} \hfill
        \end{gathered}
        \end{equation}
        \textit{is equivalent to the following reformulated optimization problem}
        \begin{equation} \label{transformed_quadratic}
        \begin{gathered}
        {\mathop{\mathrm{maximize}}\limits_{\mathbf{x}\mbox{$,$}\mathbf{y}}\hspace{0.33em}\hspace{0.33em}\hspace{0.33em}\hspace{0.33em}\mathop{\sum}\limits_{{i}\mathrm{{{}={}}}{1}}\limits^{N}{[{2}{\mathrm{Re}}{\mathrm{\{}}{y}_{i}^{\mathrm{*}}{n}_{i}{\mathrm{(}}{\mathbf{x}}{\mathrm{)\}}}\mathrm{{-}}{\mathrm{|}}{y}_{i}{\mathrm{|}}^{2}{d}_{i}{\mathrm{(}}{\mathbf{x}}{\mathrm{)]}}}} \hfill\\
        {\mathrm{subject}}\hspace{0.33em}{\mathrm{to}}\hspace{0em}\hspace{0.33em}\hspace{0.33em}\hspace{0.33em}\hspace{0.33em}{\mathbf{x}}\mathrm{\in}{\mathcal{X}}{\mathrm{,}}\hspace{0.33em}{\mathbf{y}}\mathrm{\in}{\mathbb{C}}^{N}\hfill
        \end{gathered}
        \end{equation}
        \textit{in the sense that the optimal values of the objective function and primal optimization variables are identical. Note that the vector $
                \mathrm{\mathbf{y}}=
                [{y}_{1}\hspace{0.33em} {y}_{2}\hspace{0.33em} \mathrm{\cdots}\hspace{0.33em} {y}_{\textit{N}}]$ is an auxiliary variable.}   
\end{theorem}
\begin{proof} We first observe that the reformulated objective function in (\ref{transformed_quadratic}) is concave in \textbf{y}. Thus, setting the partial derivative of this objective with respect to ${y}_{i}^{\mathrm{*}}$, we obtain
\[
{y}_{i,opt}{{}={}}\frac{{n}_{i}(\mathbf{x})}{{d}_{i}(\mathbf{x})}
\]
        Substituting these values back into the objective function in (\ref{transformed_quadratic}) yields the objective function in (\ref{quad_transform_original_objective}).
\end{proof}
Applying Theorem \ref{Thm1} to the sum-of-ratios term in ${f}_{r}{\mathrm{(}}{\mathbf{U}}{\mathbf{,}}{\mathbf{V}}{\mathbf{,}}\mathbf{\Gamma}{\mathrm{)}}$ in (\ref{f_r_single_frequency}), we obtain the following new objective function
\begin{equation} \label{f_q_LS_MIMO}
\begin{gathered}
{f}_{q}{(}{\mathbf{U}}{\mathbf{,}}{\mathbf{V}}{\mathbf{,}}\mathbf{\Gamma}{\mathbf{,}}{\mathbf{Y}}{)}{=}\mathop{\sum}\limits_{{b}{=}{1}}\limits^{\left|{\mathcal{B}}\right|}{\mathop{\sum}\limits_{{k}{=}{1}}\limits^{{K}_{b}}{\mathop{\sum}\limits_{{f}{=}{1}}\limits^{F}{\left({{w}_{kb}{\left[{\log\left({{1}{+}{\gamma}_{kb,f}}\right){-}{\gamma}_{kb,f}}\right]}{-}\mathop{\sum}\limits_{{b}{'}{=}{1}}\limits^{\left|{\mathcal{B}}\right|}{\mathop{\sum}\limits_{{k}{'}{=}{1}}\limits^{{K}_{b'}}{{u}_{k'b',f}}{\left|{{y}_{kb}}\right|}^{2}{\left({{|}{\mathbf{h}}_{kb,b',f}^{H}{\mathbf{v}}_{k'b',f}{|}^{2}{+}{\sigma}_{kb,f}^{2}}\right)}}}\right)}}}\\
{+}\mathop{\sum}\limits_{{b}{=}{1}}\limits^{\left|{\mathcal{B}}\right|}{\mathop{\sum}\limits_{{k}{=}{1}}\limits^{{K}_{b}}{\mathop{\sum}\limits_{{f}{=}{1}}\limits^{F}{2Re\{{y}_{kb,f}^{*}\sqrt{{w}_{kb}{(}{1}{+}{\gamma}_{kb,f}{)}}{u}_{kb,f}{\mathbf{v}}_{kb,f}^{H}{\mathbf{h}}_{kb,b,f}\}}}}
\end{gathered}
\end{equation}
and accordingly, (\ref{first_reformulated_problem_LS_MIMO}) can be expressed as the equivalent optimization problem below
\begin{subequations} \label{final_reformulated_problem_LS_MIMO}
        \begin{align}
        \mathop{\mathrm{maximize}}\limits_{{\mathbf{U}}{\mathrm{,}}{\mathbf{V}}{\mathbf{,}}\mathbf{\Gamma}{\mathbf{,}}{\mathbf{Y}}}\hspace{2.30em}{f}_{q}{\mathrm{(}}{\mathbf{U}}{\mathbf{,}}{\mathbf{V}}{\mathbf{,}}\mathbf{\Gamma}{\mathbf{,}}{\mathbf{Y}}{\mathrm{)}}\hspace{22.83em}\\
        {\mathrm{subject}}\hspace{0.33em}{\mathrm{to}}\hspace{2.20em}{u}_{kb\mathrm{,}f}\mathrm{\in}{\mathrm{\{}}{0}{\mathrm{,}}{1}{\mathrm{\},}}\hspace{5.69em}{b}\mathrm{{{}={}}}{1}{\mathrm{,...,}}\hspace{0.33em}\left|{\mathcal{B}}\right|{\mathrm{;}}\hspace{0.33em}{k}\mathrm{{{}={}}}{1}{\mathrm{,...,}}\hspace{0.33em}{K}_{b}{\mathrm{;}}\hspace{0.33em}{f}\mathrm{{{}={}}}{1}{\mathrm{,...,}}\hspace{0.33em}{F}\hspace{0.53em}\\
        \mathop{\sum}\limits_{{k}\mathrm{{{}={}}}{1}}\limits^{{K}_{b}}{{u}_{kb\mathrm{,}f}}\hspace{0.33em}\mathrm{\leq}\hspace{0.33em}{M}{\mathrm{,}}\hspace{4.54em}{b}\mathrm{{{}={}}}{1}{\mathrm{,...,}}\hspace{0.33em}\left|{\mathcal{B}}\right|{\mathrm{;}}\hspace{0.33em}{f}\mathrm{{{}={}}}{1}{\mathrm{,...,}}\hspace{0.33em}{F}\hspace{6.45em}\\
        \mathop{\sum}\limits_{{f}\mathrm{{{}={}}}{1}}\limits^{F}{\mathop{\sum}\limits_{{k}\mathrm{{{}={}}}{1}}\limits^{{K}_{b}}{{\left\|{{\mathbf{v}}_{kb}}\right\|}_{2}^{2}\hspace{0.33em}\mathrm{\leq}\hspace{0.33em}{F}{P}_{T}}}{\mathrm{,}}\hspace{1.14em}{b}\mathrm{{{}={}}}{1}{\mathrm{,...,}}\hspace{0.33em}\left|{\mathcal{B}}\right|\hspace{12.23em}\label{inequality_constraint_final_problem_LS_MIMO}
        \end{align}
\end{subequations}
Once again, to avoid excessive notational clutter, we collect the ${y}_{kb\mathrm{,}f}$ values in the matrix $\mathbf{Y}$.

\indent From Theorem \ref{Thm1}, it follows that the optimization problem in (\ref{final_reformulated_problem_LS_MIMO}) is also equivalent to the original optimization problem in (\ref{LS_MIMO_equiv_problem}) in the sense that the optimal objective function and the associated primal optimization variables, $\mathbf{U}$ and $\mathbf{V}$, for both problems are identical.

\indent We emphasize that both of the reformulated problems in (\ref{first_reformulated_problem_LS_MIMO}) and (\ref{final_reformulated_problem_LS_MIMO}) remain nonconvex and NP-hard (like the original problem) since our reformulation steps result in equivalent problems. Crucially, however, the new objective functions are now in a form amenable to an iterative optimization strategy leading to an effective solution to our original optimization problem in (\ref{LS_MIMO_equiv_problem}). 

\noindent \emph{The continuous variables:} To develop the iterative approach, we first observe that for fixed $\mathbf{U}, \mathbf{V}$ and $\mathbf{Y}$, the optimal $\mathbf{\Gamma}$ can be found by setting 
\begin{equation} \label{gamma_kbs_opt}
\frac{\mathrm{\partial}{f}_{q}{\mathrm{(}}{\mathbf{U}}{\mathbf{,}}{\mathbf{V}}{\mathbf{,}}\mathbf{\Gamma}{\mathbf{,}}{\mathbf{Y}}{\mathrm{)}}}{\mathrm{\partial}{\mathit{\gamma}}_{kb\mathrm{,}f}}\mathrm{{{}={}}}{0}\hspace{1.33em}\mathrm{\Rightarrow}\hspace{1.33em}{\mathit{\gamma}}_{kb\mathrm{,}f\mathrm{,}opt}\mathrm{{{}={}}}\frac{{u}_{kb\mathrm{,}f}\mathrm{|}{\mathbf{h}}_{kb\mathrm{,}b\mathrm{,}f}^{H}{\mathbf{v}}_{kb\mathrm{,}f}{\mathrm{|}}^{2}}{\mathop{\sum}\limits_{\mathop{{b}{\mathrm{'}}\mathrm{{{}={}}}{1}}\limits_{{\mathrm{(}}{b}\prime{\mathrm{,}}{k}\prime{\mathrm{)}}\kern0.2em\mathrm{\ne}}}\limits^{\left|{\mathcal{B}}\right|\hspace{0.33em}}{\mathop{\sum}\limits_{\mathop{{k}{\mathrm{'}}\mathrm{{{}={}}}{1}}\limits_{\mathrm{(}b\mathrm{,}k\mathrm{)}}}\limits^{{K}_{b\mathrm{'}}}{{u}_{k\mathrm{'}b\mathrm{',}f}{\mathrm{|}}{\mathbf{h}}_{kb\mathrm{,}b\mathrm{',}f}^{H}{\mathbf{v}}_{k\mathrm{'}b\mathrm{',}f}{\mathrm{|}}^{2}\mathrm{{+}}{\mathit{\sigma}}_{kb\mathrm{,}f}^{2}}}}
\end{equation}
as $
{f}_{r}\left({\mathbf{U}{,}\mathbf{V}{,}\mathbf{\Gamma}}\right)$ is concave in $\mathbf{\Gamma}$. 
Next, we note that holding the $\mathbf{U}$, $\mathbf{V}$ and $\mathbf{\Gamma}$ values fixed, $
{f}_{q}\left({{\mathbf{U}}{\mathbf{,}}{\mathbf{V}}{\mathbf{,}}\mathbf{\Gamma}{\mathbf{,}}{\mathbf{Y}}}\right)
$ is concave in $\mathbf{Y}$, so the optimal $\mathbf{Y}$ values can be found by setting
\begin{equation} \label{y_kbs_opt}
\frac{\mathrm{\partial}{f}_{q}{\mathrm{(}}{\mathbf{U}}{\mathbf{,}}{\mathbf{V}}{\mathbf{,}}\mathbf{\Gamma}{\mathbf{,}}{\mathbf{Y}}{\mathrm{)}}}{\mathrm{\partial}{y}_{kb\mathrm{,}f}^{\mathrm{*}}}\mathrm{{{}={}}}{0}\hspace{1.33em}\mathrm{\Rightarrow}\hspace{1.33em}{y}_{kb\mathrm{,}f\mathrm{,}opt}\mathrm{{{}={}}}\frac{\sqrt{{w}_{kb}{\mathrm{(}}{1}\mathrm{{+}}{\mathit{\gamma}}_{kb\mathrm{,}f}{\mathrm{)}}}{u}_{kb\mathrm{,}f}{\mathbf{v}}_{\mathit{k}\mathit{b}\mathrm{,}\mathit{f}}^{\mathit{H}}{\mathbf{h}}_{\textit{kb}\mathrm{,}\textit{b}\mathrm{,}\textit{f}}}{\mathop{\sum}\limits_{{b}{\mathrm{'}}\mathrm{{{}={}}}{1}}\limits^{\left|{\mathcal{B}}\right|}{\mathop{\sum}\limits_{{k}{\mathrm{'}}\mathrm{{{}={}}}{1}}\limits^{{K}_{b\mathrm{'}}}{{u}_{k\mathrm{'}b\mathrm{',}f}{\mathrm{|}}{\mathbf{h}}_{kb\mathrm{,}b\mathrm{',}f}^{H}{\mathbf{v}}_{k\mathrm{'}b\mathrm{',}f}{\mathrm{|}}^{2}\mathrm{{+}}{\mathit{\sigma}}_{kb\mathrm{,}f}^{2}}}}
\end{equation}
{
In a similar fashion, when  $\mathbf{U}$, $\mathbf{\Gamma}$ and $\mathbf{Y}$ are fixed, we can find the optimal $\mathbf{V}$ values (i.e., the beamforming weight vectors). Note that due to the sum-power constraint (\ref{inequality_constraint_final_problem_LS_MIMO}), taking the derivative of ${f}_{q}\left({{\mathbf{U}}{\mathbf{,}}{\mathbf{V}}{\mathbf{,}}\mathbf{\Gamma}{\mathbf{,}}{\mathbf{Y}}}\right)$ directly with respect to $\mathbf{V}$ to find the optimal beamforming weight vectors is not valid. To simplify the derivation, we recall that with these variables fixed, we can write the problem of finding the optimal beamforming weight vectors as:
\begin{subequations} \label{BF_weights_problem_LS_MIMO}
        \begin{align}
        \mathop{\mathrm{maximize}}\limits_{{\mathbf{V}}}\hspace{2.30em}{f}_{q}{\mathrm{(}}{\mathbf{U}}{\mathbf{,}}{\mathbf{V}}{\mathbf{,}}\mathbf{\Gamma}{\mathbf{,}}{\mathbf{Y}}{\mathrm{)}}\hspace{22.83em}\\
        {\mathrm{subject}}\hspace{0.33em}{\mathrm{to}}\hspace{2.30em}\mathop{\sum}\limits_{{f}\mathrm{{{}={}}}{1}}\limits^{F}{\mathop{\sum}\limits_{{k}\mathrm{{{}={}}}{1}}\limits^{{K}_{b}}{{\left\|{{\mathbf{v}}_{kb}}\right\|}_{2}^{2}\hspace{0.33em}\mathrm{\leq}\hspace{0.33em}{F}{P}_{T}}}{\mathrm{,}}\hspace{1.14em}{b}\mathrm{{{}={}}}{1}{\mathrm{,...,}}\hspace{0.33em}\left|{\mathcal{B}}\right|\hspace{12.23em}\label{BF_weights_constraint}
        \end{align}
\end{subequations}
We note that this optimization problem is concave and thus readily solvable; in fact, we can derive a closed-form expression for the weights by introducing Lagrange multipliers $\mu_{b}$ for the sum-power constraint at each BS in (\ref{BF_weights_constraint}). This yields the Lagrangian $\hat{\mathcal{L}}\left({{\mathbf{V}}{\mathbf{,}}\boldsymbol{{\mu}}}\right)$ (not to be confused with the Lagrangian ${\mathcal{L}}{(}{\mathbf{U}}{\mathbf{,}}{\mathbf{V}}{\mathbf{,}}\mathbf{\Gamma}{\mathbf{,}}\mathbf{\Lambda}{)}$
we derived earlier) as:
\begin{equation}
\hat{\mathcal{L}}\left({{\mathbf{V}}{\mathbf{,}}\boldsymbol{\mu}}\right){=}{f}_{q}{(}{\mathbf{U}}{\mathbf{,}}{\mathbf{V}}{\mathbf{,}}\mathbf{\Gamma}{\mathbf{,}}{\mathbf{Y}}{)}{-}\mathop{\sum}\limits_{{b}{=}{1}}\limits^{\left|{\mathcal{B}}\right|}{{\mathit{\mu}}_{b}\left({\mathop{\sum}\limits_{{f}{=}{1}}\limits^{F}{\mathop{\sum}\limits_{{k}{=}{1}}\limits^{{K}_{b}}{{\left\|{{\mathbf{v}}_{kb,f}}\right\|}_{2}^{2}\hspace{0.33em}{-}{FP}_{T}}}}\right)}
\end{equation}
where, to keep our notation uncluttered, we collect the multipliers $\mu_{b}$ in the vector $\boldsymbol{\mu}$. Then the first-order optimality condition of $\hat{\mathcal{L}}\left({{\mathbf{V}}{\mathbf{,}}\boldsymbol{{\mu}}}\right)$ with respect to each $\mathbf{v}_{kb,f}$ yields:
\begin{multline} \label{v_kbs_opt}
\frac{\partial\hat{\mathcal{L}}\left({{\mathbf{V}}{\mathbf{,}}\boldsymbol{\mu}}\right)}{\partial{\mathbf{v}}_{kb,f}}{=}\mathbf{0}\Rightarrow\\\mathbf{v}_{kb,f,opt}{=}\sqrt{{w}_{kb}{(}{1}{+}{\gamma}_{kb,f}{)}}{u}_{kb,f}{\left({\mathop{\sum}\limits_{{b}{'}{=}{1}}\limits^{B}{\mathop{\sum}\limits_{{k}{=}{1}}\limits^{{K}_{b}}{{u}_{k'b',f}{\left|{{y}_{k'b',f}}\right|}^{2}{\mathbf{h}}_{k'b',b,f}}}{\mathbf{h}}_{k'b',b'f}^{H}{+}{\mu}_{b}{\mathbf{I}}_{M}}\right)}^{{-}{1}}{y}_{kb,f}^{*}{\mathbf{h}}_{kb,b,f}
\end{multline}
The dual variable ${{\mu}_{b}}$ should be chosen to satisfy complementary slackness in the total power constraint at BS $b$; observing (\ref{v_kbs_opt}), it is clear that the magnitude of $\mathbf{v}_{kb,f}$ is a decreasing function of ${{\mu}_{b}}$. Thus, we can obtain ${{\mu}_{b}}$ easily through a bisection search, which in turn can be used to obtain the optimal beamforming weight vectors.
}
At this juncture, we note an important point for future reference: the beamforming step involves inversion of a $M\times{M}$ matrix \textit{for each user}, which is computationally costly, especially when we have to perform it for a large number of users. Furthermore, if the scheduling variable for a user is zero, the beamforming weight for that user is automatically zero; there is no need to perform any computation in this case.

\noindent \emph{The binary scheduling variables:} The final step in the iterative approach is to optimize the user scheduling variables $\mathbf{U}$ when the continuous variables, $\mathbf{\Gamma}$, $\mathbf{V}$ and $\mathbf{Y}$, are held fixed. To do so, we first observe that since the frequency bands are assumed to be orthogonal, the scheduling decisions are decoupled across the different frequency bands. Furthermore, we make use of an intuitive yet powerful insight first suggested in~\cite{stolyar2009self} and also observed in \cite{yu2013multicell}: provided that the beamforming weight vectors $\mathbf{V}$ are held fixed, the interference value experienced by a user in the downlink scheduled on a particular frequency band \textit{depends only on the beam used to serve that user and remains fixed regardless of which other users are scheduled on the remaining beams}. We note that this is different from the uplink setting, in which the interference pattern in the network changes when a new set of users is scheduled on a given set of beams.

The fact that the interference pattern changes in the uplink creates significant challenges in terms of scheduling: as the authors of \cite{shen2018fractional2} observe, even with beamforming weights \textit{fixed}, the problem of optimal scheduling remains NP-hard. This substantially affects the quality of solutions obtained since as emphasized earlier, we do not know which users are suitable to schedule \textit{a priori}. One solution, as mentioned previously, is to schedule all users in the network and transform the original problem to an unconstrained problem in terms of scheduling; however, since we need a matrix inversion per user, this results in an undesirable increase in computational complexity to levels identical with the WMMSE algorithm \cite{shen2018fractional2}. In the downlink, however, we are not bound by this constraint, i.e., changing scheduling decisions does not affect the interference pattern. Thus, we can schedule only a subset of users in the entire network, reducing complexity as only the beamforming weights for a small number of users need to be calculated.

To illustrate this point, let us consider the ${b}^{\mathrm{th}}$ BS in a multicell network. The scheduling decisions for different frequency bands are decoupled; thus, we consider the ${f}^{\mathrm{th}}$ frequency band without loss of generality. In addition to this, each BS can schedule at most $M$ users in a single time slot, whereas it has $K_b\mathrm{>}M$ users associated with it. Suppose the BS is serving users on the indicated frequency band using a fixed set of ${N_{b,f}}\mathrm{\leq}{M}$ nonzero beamforming weights which we denote by ${\check{\mathcal{V}}}_{b,f}$, i.e., \[{\check{\mathcal{V}}}_{b,f}\mathrm{{{}={}}}{\mathrm{\{}}{\mathbf{\check{v}}}_{nb,f}\mathrm{\in}{\mathbb{C}}^{M}{\mathrm{|}}\hspace{0.33em}{n}\mathrm{{{}={}}}{1}{\mathrm{,...,}}{N}_{b,f}{\mathrm{;}}\hspace{0.33em}{\mathbf{\check{v}}}_{nb,f}\mathrm{\ne}{\mathbf{0}}{\mathrm{\}}}\]
\indent It follows that the ${k}^{\mathrm{th}}$ user associated with this BS can find itself in one of two scenarios with regards to the given frequency band: either it is scheduled for transmission on the one of the ${N_{b,f}}$ beams from the set ${\check{\mathcal{V}}}_{b,f}$ or it is not being served by the BS. For the former setting, suppose the user is scheduled on the ${n}^{\mathrm{th}}$ nonzero beam; then the power received on this beam is the signal power. If the user is not scheduled, however, all the power received is interference power. For notational convenience, let us denote by ${\mathit{\zeta}}_{kb\mathrm{,}f}$ the combined received signal, interference and noise power of the user in question, i.e.,
\begin{equation}\label{intf_plus_noise}
{\zeta}_{kb,f}{=}\mathop{\sum}\limits_{{b}{'}{=}{1}}\limits^{\left|{\mathcal{B}}\right|}{\mathop{\sum}\limits_{{k}{'}{=}{1}}\limits^{K}{{u}_{k'b',f}{\left|{{\mathbf{h}}_{kb,b',f}^{\mathrm{H}}{\mathbf{v}}_{k'b',f}}\right|}^{2}{+}{\mathit{\sigma}}_{kb,f}^{2}}}
\end{equation}

Then the total interference power received by this user ${I}_{kb,f}$ is given by 
\[
{I}_{kb,f}{{}={}}\left\{{\begin{array}{c}{{\mathit{\zeta}}_{kb,f}\qquad\qquad\qquad\qquad\qquad{\mathrm{if}}\hspace{0.33em}{\mathrm{the}}\hspace{0.33em}\mathrm{user}\hspace{0.33em}\mathrm{is}\hspace{0.33em}\mathrm{not}\hspace{0.33em}\mathrm{scheduled}\mathrm{,i.e.,}\hspace{0.33em}{u}_{kb,f}{{}={}}{0}{}}\\{{\mathit{\zeta}}_{kb,f}{-}{\left|{{\mathbf{h}}_{kb,b,f}^{H}{\check{\mathbf{v}}}_{nb,f}}\right|}^{2}\hspace{1.73em}\qquad\mathrm{if}\hspace{0.33em}\mathrm{the}\hspace{0.33em}\mathrm{user}\hspace{0.33em}\mathrm{is}\hspace{0.33em}\mathrm{scheduled}\mathrm{,i.e.,}\hspace{0.33em}{u}_{kb,f}{{}={}}{1}{}\hspace{1.76em}}\end{array}}\right.
\]

Importantly, if this user is scheduled on the ${n}^{\mathrm{th}}$ nonzero beam, the interference power it experiences does not depend on which users are scheduled on the remaining ${N}_{b,f}\mathrm{{-}}{1}$ beams within its own cell. The same holds true when the user in question is not scheduled on any beam; the interference power experienced by the user is the same regardless of which set of users in its own cell is scheduled on the given set of beams. In addition, we also make the following critical observation: for the ${b}^{\mathrm{th}}$ BS, changing the set of users scheduled on the fixed set of beams does not change the interference pattern experienced by users \textit{outside} its own cell. {This can be seen from the fact that in (\ref{intf_plus_noise}), the inter-cell interference power received by user $k$ associated with BS $b$ from BS $b'$ on frequency band $f$ depends only on the interference channel $\mathbf{h}_{kb,b',f}$, and the beamforming weight vectors $\mathbf{v}_{k'b',f}$ can be permuted over any of the users served by BS $b'$ without affecting ${\zeta}_{kb,f}$. Meanwhile, the intracell interference power recieved by user $k$ associated with BS $b$ on frequency band $f$ depends only upon the information-bearing channel $\mathbf{h}_{kb,b,f}$ and the beamforming weight the user is scheduled on $\mathbf{v}_{k'b',f}$.  Taken together, these observations imply that if the beamforming weights throughout the network are held fixed, we can \textit{locally} optimize the scheduling at each BS in order to maximize the \textit{network-wide} sum weighted rate \emph{for the given set of beamforming weights}.} 

Accordingly, we can formulate a strategy to help us find the best set of ${N}_{b,f}$ users to be served by the ${b}^{\mathrm{th}}$ BS on its set of nonzero beamforming weights ${\check{\mathcal{V}}}_{b,f}$ for the ${f}^{\mathrm{th}}$ frequency band. In other words, our goal is to find the set of ${N}_{b,f}$ users out of the ${K}_{b}$ total users in the cell that will yield the maximum weighted sum rate on the given set of beamforming weight vectors. Our choice of users should satisfy the constraint that a user can only be served on a single beam by a BS in keeping with our original system model. A greedy strategy of assigning the user capable of achieving the highest weighted rate on each beam is not guaranteed to solve the combinatorial problem of selecting the best subset of ${N}_{b,f}$ users of the $K_b$ available. 

Our first step in matching the users to the fixed beams to maximize the WSR is to define the ${K}_{b}\mathrm{\times}{N}_{b,f}$ combined weighted rates matrix ${\hat{\mathbf{R}}}_{b\mathrm{,}f}$ for the given set of beams and users on the ${f}^{\mathrm{th}}$ frequency band. The $\mathrm{(}i\mathrm{,}j{\mathrm{)}}^{\mathrm{th}}$ entry in this matrix, denoted by ${\widehat{r}}_{ib\mathrm{,}j\mathrm{,}f}$, indicates the weighted rate that would be achieved by the ${i}^{\mathrm{th}}$ user if it is scheduled on the ${j}^{\mathrm{th}}$ nonzero beam on the ${f}^{\mathrm{th}}$ frequency band, i.e.,
\[
{\left[{{\hat{\mathbf{R}}}_{b\mathrm{,}f}}\right]}_{ij}\mathrm{{{}={}}}{\widehat{r}}_{ib\mathrm{,}j\mathrm{,}f}\mathrm{{{}={}}}{w}_{ib}\log\left({{1}\mathrm{{+}}\frac{{\left|{{\mathbf{h}}_{ib\mathrm{,}b\mathrm{,}f}^{H}{\check{\mathbf{v}}}_{jb\mathrm{,}f}}\right|}^{2}}{{\mathit{\zeta}}_{ib\mathrm{,}f}\mathrm{{-}}{\left|{{\mathbf{h}}_{ib\mathrm{,}b\mathrm{,}f}^{H}{\check{\mathbf{v}}}_{jb\mathrm{,}f}}\right|}^{2}}}\right)
\]
Note that we compute the total interference received by every user in the network, regardless of whether it is scheduled or not; thus, every user is considered eligible for possible scheduling on a non-zero beamforming weight.

It follows that our goal of scheduling the users on the appropriate beams can be formulated as the following binary integer optimization problem:
\begin{equation} \label{LSAP_single_freq}
\begin{gathered}
{\mathop{\mathrm{maximize}}\limits_{\mathbf{X}}\hspace{0.33em}\hspace{0.33em}\hspace{0.33em}\hspace{0.53em}\hspace{0.33em}\hspace{0.33em}\mathop{\sum}\limits_{{k}\mathrm{{{}={}}}{1}}\limits^{{K}_{b}}{\mathop{\sum}\limits_{{n}\mathrm{{{}={}}}{1}}\limits^{{N}_{b\mathrm{,}f}}{{\widehat{r}}_{kb\mathrm{,}n\mathrm{,}f}}}{x}_{kb\mathrm{,}n\mathrm{,}f}}\hspace{12.65em}\\
{{\mathrm{subject}}\hspace{0.33em}{\mathrm{to}}\hspace{0.20em}\hspace{0.33em}\hspace{0.33em}\hspace{0.33em}\hspace{0.33em}\hspace{0.33em}\mathop{\sum}\limits_{{k}\mathrm{{{}={}}}{1}}\limits^{{K}_{b}}{{x}_{kb\mathrm{,}n\mathrm{,}f}}\hspace{0.33em}\mathrm{{{}={}}}\hspace{0.33em}{1}{\mathrm{,}}\hspace{2.52em}{n}\mathrm{{{}={}}}{1}{\mathrm{,...,}}{N}_{b\mathrm{,}f}}\hspace{5.95em}\\
{\mathop{\sum}\limits_{{n}\mathrm{{{}={}}}{1}}\limits^{{N}_{b\mathrm{,}f}}{{x}_{kb\mathrm{,}n\mathrm{,}f}\hspace{0.33em}\mathrm{\leq}\hspace{0.33em}{1}}{\mathrm{,}}\hspace{1.00em}\hspace{0.33em}\hspace{0.33em}\hspace{0.33em}\hspace{0.33em}\hspace{0.33em}\hspace{0.33em}{k}\mathrm{{{}={}}}{1}{\mathrm{,...,}}{K}_{b}}\hspace{0.30em}\\
{{x}_{kb\mathrm{,}n\mathrm{,}f}\mathrm{\in}{\mathrm{\{}}{0}{\mathrm{,}}{1}{\mathrm{\},}}\hspace{1.95em}\hspace{0.33em}\hspace{0.33em}\hspace{0.33em}\hspace{0.33em}{n}\mathrm{{{}={}}}{1}{\mathrm{,...,}}{N}_{b\mathrm{,}f}{\mathrm{;}}\hspace{0.33em}{k}\mathrm{{{}={}}}{1}{\mathrm{,...,}}{K}_{b}\hspace{-6.10em}}
\end{gathered}
\end{equation}
where the binary variables ${x}_{kb\mathrm{,}n\mathrm{,}f}$ indicate whether or not the ${k}^{\mathrm{th}}$ user is scheduled on the ${n}^{\mathrm{th}}$ nonzero beam by the ${b}^{\mathrm{th}}$ BS on the ${f}^{\mathrm{th}}$ frequency band. The objective function maximizes the WSR. {The first constraint requires at least one user to be scheduled on every beam while the second ensures that a user is scheduled in one beam only.} 

Note that these binary variables are not the same as the optimization variables $u_{kb,f}$ which refer to whether the user is scheduled or not. These variables have the additional index $n$ which denotes the $n^\mathrm{th}$ beam. Specifically, $u_{kb,f} = \sum_{n=1}^{N_{b,f}} x_{kb,n,f}$, i.e., if the user is scheduled on any one of the beams associated with the BS, it is scheduled by that BS.

The problem in~\eqref{LSAP_single_freq} is, in fact, a linear sum assignment problem, and can also be viewed as a maximum weighted bipartite matching problem, which has been extensively studied in the literature and can be solved in polynomial time using techniques like the Hungarian algorithm~\cite{grotschel2012geometric} (more formally known as the Kuhn-Munkres algorithm) or the auction algorithm~\cite{bertsekas1990auction}. Specifically, using the Hungarian algorithm to solve the linear sum assignment problem for an ${m}\mathrm{\times}{n}$ matrix, where $m\mathrm{>}n$, has a complexity of $\mathcal{O}\mathrm{(}{n}^{2}m\mathrm{)}$~\cite{grotschel2012geometric}. In our case, we have ${N}_{b,f}\mathrm{<}{K}_{b}$; hence the complexity of solving problem (\ref{LSAP_single_freq}) is $\mathcal{O}\mathrm{(}{N}_{b,f}^{2}{K}_{b}\mathrm{)}$. Solving this optimization problem for each BS allows us to \emph{optimally} schedule the users to maximize the network weighted sum rate on the fixed set of beamforming weight vectors. We remark that this scheduling scheme using fractional programming and Hungarian algorithm is different from the uplink setting \cite{shen2018fractional2}, where the scheduling of one user would have changed the interference pattern; consequently the only way to solve the uplink problem to global optimality is by extensive search as mentioned earlier.

This scheduling setup also reduces complexity as compared to the unconstrained scheduling setting, as we now only have to calculate the beamforming weight vectors for a maximum of $M$ rather than $K_b$ users at each iteration of the algorithm. Importantly, this set of $M$ users can change from iteration to iteration as the beamforming weights are matched to the best set of users; this is unlike the uplink setting where a user not scheduled during the initialization remains unscheduled throughout all subsequent iterations of fractional programming algorithm \cite{shen2018fractional2}. 

{With a fixed set of beamforming weight vectors, therefore, the proposed scheduling scheme finds a per-cell optimal selection of scheduling decisions; hence, from iteration to iteration, the user assignment to beamforming weight vectors is set according to which combination yields the greatest network WSR. Here, we only consider scheduling each user on a single beam per frequency band for two reasons: first, this is in keeping with our original system model in which each user is scheduled on a maximum of one data stream per frequency band and the standard assumption for coordinated resource allocation algorithms including multicell WMMSE \cite{shi2011iteratively} and uplink fractional programming \cite{shen2018fractional2}, and second, the framework of the Hungarian algorithm does not allow for scheduling on more than a single beamforming weight vector.}

Combining all these steps together, the proposed technique for coordinated resource allocation in the downlink of multiuser MISO networks is summarized in Algorithm 1. The algorithm optimizes one of the optimization variables keeping the others fixed, iterating till convergence.

\begin{algorithm}
        \caption{Coordinated Resource Allocation for Downlink of LS-MIMO Networks}\label{euclid}
        \begin{algorithmic}[1]
                \State ${\mathrm{Initialize}}\hspace{0.33em}{\mathbf{U}}{\mathrm{,}}\hspace{0.33em}{\mathbf{V}}{\mathrm{,}}\hspace{0.33em}{N}_{iterations}$
                \State{$\mathrm{Set\hspace{0.33em}\mathit{i}=1}$}
                \State {\textbf{repeat}}
                \State{$\hspace{1.45em}{\mathrm{Update}}\hspace{0.33em}\mathbf{\Gamma}\hspace{0.33em}{\mathrm{using}}\hspace{0.33em}(\ref{gamma_kbs_opt})$}
                \State{$\hspace{1.45em}{\mathrm{Update}}\hspace{0.33em}\mathbf{Y}\hspace{0.33em}{\mathrm{using}}\hspace{0.33em}(\ref{y_kbs_opt})$}
                \State{$\hspace{1.45em}{\mathrm{Update}}\hspace{0.33em}\mathbf{V}\hspace{0.33em}{\mathrm{using}}\hspace{0.33em}(\ref{v_kbs_opt})$}
                \State{$\hspace{1.45em}{\mathrm{Update}}\hspace{0.33em}{\mathbf{U}}\hspace{0.33em}{\mathrm{and}}\hspace{0.33em}{\mathbf{V}}\hspace{0.33em}{\mathrm{jointly}}\hspace{0.33em}{\mathrm{by}}\hspace{0.33em}{\mathrm{solving}}\hspace{0.33em}(\ref{LSAP_single_freq})\hspace{0.33em}{\mathrm{for}}\hspace{0.33em}{\mathrm{each}}\hspace{0.33em}{\mathrm{base}}\hspace{0.33em}{\mathrm{station}}\hspace{0.33em}{\mathrm{on}}\hspace{0.33em}{\mathrm{each}}\hspace{0.33em}{\mathrm{frequency}}\hspace{0.33em}{\mathrm{band}}$}
                \State{\hspace{1.30em} $\mathrm{Increment\hspace{0.33em}\mathit{i}}$}
                \State {${\mathbf{until}}\hspace{0.33em}{\mathrm{convergence}}\hspace{0.33em}{\mathrm{or}}\hspace{0.33em}{i}\mathrm{{{}={}}}{N}_{iterations}$}
        \end{algorithmic}
\end{algorithm}
\newpage
\begin{theorem}
        The proposed algorithm described in Algorithm 1 is non-decreasing in the objective function ${f}_{0}\mathrm{(}\mathbf{U}\mathrm{,}\mathbf{V}\mathrm{)}$ after each iteration.
\end{theorem}
\textit{Proof:} We refer to the objective function in~\eqref{aux_objective_LS_MIMO} as ${f}_{0}\mathrm{(}\mathbf{U}\mathrm{,}\mathbf{V}\mathrm{)}$. The non-decreasing convergence can be proven by considering the following chain of reasoning going from iteration $i$ to $i+1$:
\begin{subequations}
        \begin{align}
        {f}_{0}{\mathrm{(}}{\mathbf{U}}^{\mathrm{(}i\mathrm{)}}{\mathrm{,}}{\mathbf{V}}^{\mathrm{(}i\mathrm{)}}{\mathrm{)}}\mathrm{{{}={}}}{f}_{r}{\mathrm{(}}{\mathbf{U}}^{\mathrm{(}i\mathrm{)}}{\mathrm{,}}{\mathbf{V}}^{\mathrm{(}i\mathrm{)}}{\mathrm{,}}{\mathbf{\Gamma}}^{\mathrm{(}i\mathrm{)}}{\mathrm{)}}\label{proof_line_1}\hspace{4.60em}\\
        \leq{f}_{r}({{\mathbf{U}}^{\left({i}\right)},{\mathbf{V}}^{\left({i}\right)},{\mathbf{\Gamma}}^{\left({{i}{+}{1}}\right)}})\label{proof_line_2}\hspace{3.4em}\\
        {{}={}}{f}_{q}({{\mathbf{U}}^{\left({i}\right)},{\mathbf{V}}^{\left({i}\right)},{\mathbf{\Gamma}}^{\left({{i}{+}{1}}\right)},{\mathbf{Y}}^{\left({i}\right)}})\label{proof_line_3}\hspace{1.25em}\\
        \mathrm{\leq}{f}_{q}{\mathrm{(}}{\mathbf{U}}^{\mathrm{(}i\mathrm{)}}{\mathrm{,}}{\mathbf{V}}^{\mathrm{(}i\mathrm{)}}{\mathrm{,}}{\mathbf{\Gamma}}^{{\mathrm{(}}{i}\mathrm{{+}}{1}{\mathrm{)}}}{\mathrm{,}}{\mathbf{Y}}^{{\mathrm{(}}{i}\mathrm{{+}}{1}{\mathrm{)}}}{\mathrm{)}}\hspace{1.15em}\label{proof_line_4}\\    
        \mathrm{\leq}{f}_{q}{\mathrm{(}}{\mathbf{U}}^{\mathrm{(}i\mathrm{)}}{\mathrm{,}}{\mathbf{V}}^{{\mathrm{(}}{i}\mathrm{{+}}{1}{\mathrm{)}}}{\mathrm{,}}{\mathbf{\Gamma}}^{{\mathrm{(}}{i}\mathrm{{+}}{1}{\mathrm{)}}}{\mathrm{,}}{\mathbf{Y}}^{{\mathrm{(}}{i}\mathrm{{+}}{1}{\mathrm{)}}}{\mathrm{)}}\hspace{0.18em}\label{proof_line_5}\\
        \mathrm{\leq}{f}_{q}{\mathrm{(}}{\mathbf{U}}^{{\mathrm{(}}{i}\mathrm{{+}}{1}{\mathrm{)}}}{\mathrm{,}}{\tilde{\mathbf{V}}}^{{\mathrm{(}}{i}\mathrm{{+}}{1}{\mathrm{)}}}{\mathrm{,}}{\mathbf{\Gamma}}^{{\mathrm{(}}{i}\mathrm{{+}}{1}{\mathrm{)}}}{\mathrm{,}}{\mathbf{Y}}^{{\mathrm{(}}{i}\mathrm{{+}}{1}{\mathrm{)}}}{\mathrm{)}}\hspace{-0.70em}\label{proof_line_6}\\
        {{}={}}{f}_{r}\left({{\mathbf{U}}^{\left({{i}{+}{1}}\right)},{\tilde{\mathbf{V}}}^{\left({{i}{+}{1}}\right)},{\mathbf{\Gamma}}^{\left({{i}{+}{1}}\right)}}\right)\hspace{0.95em}\label{proof_line_7}\\
        \mathrm{{{}={}}}{f}_{0}{\mathrm{(}}{\mathbf{U}}^{{\mathrm{(}}{i}\mathrm{{+}}{1}{\mathrm{)}}}{\mathrm{,}}{\tilde{\mathbf{V}}}^{{\mathrm{(}}{i}\mathrm{{+}}{1}{\mathrm{)}}}{\mathrm{)}}\label{proof_line_8}\hspace{4.50em}
        \end{align}
\end{subequations}
where (\ref{proof_line_1}) follows from the fact that the reformulated objective function $f_r$ equals the original when the optimal $\mathbf{\Gamma}$
values are substituted; (\ref{proof_line_2}) follows from the fact that the update of $\mathbf{\Gamma}$ when all other variables are fixed maximizes $f_r$; (\ref{proof_line_3}) follows from Theorem 1; (\ref{proof_line_4}) follows from the fact that the update of $\mathbf{Y}$ when all other variables are fixed maximizes $f_q$; (\ref{proof_line_5}) follows from the fact that the update of $\mathbf{V}$ when all other variables are fixed maximizes $f_q$; (\ref{proof_line_6}) follows from the fact that the joint update of $\mathbf{U}$ and $\mathbf{V}$ using (\ref{LSAP_single_freq}) maximizes $f_q$ when all other variables are fixed; (\ref{proof_line_7}) follows from Theorem 1; and (\ref{proof_line_8}) from similar reasoning to (\ref{proof_line_1}). Note that we use $\tilde{\mathbf{V}}$ to denote the set of permuted beamforming weights obtained from $\mathbf{V}$ by solving (\ref{LSAP_single_freq}). $\blacksquare$

Coupled with the fact that the objective function has a finite maximum, we can state that the algorithm converges. However, since the scheduling variables are binary, we cannot call this a local optimum. Furthermore, we observe that the proposed scheme is not exactly a block coordinate ascent scheme, since we use the partial derivative of $f_r$ in (\ref{gamma_kbs_opt}); nonetheless, the algorithm converges in a non-decreasing fashion to an effective solution of the original WSR maximization problem as described above.

\section{Performance Evaluation of Proposed Scheme} \label{ss results single freq}

In order to evaluate the performance of the proposed algorithm, we compare it with the following different coordinated and uncoordinated resource allocation schemes:
\begin{enumerate}
        \item \textit{Matched filtering transmission with equal power allocation and round-robin scheduling}: This is the simplest uncoordinated resource allocation strategy which can be implemented and as such it provides a useful benchmark with which to compare the performance of the multiuser algorithm.
        \item \textit{Zero-forcing with equal power allocation and round-robin scheduling}: Zero-forcing eliminates intracell interference and thus provides improved performance compared to matched filtering. Zero-forcing does involve increased computational complexity compared to matched filtering, requiring a matrix inversion to determine the beamforming weight for each of the scheduled users. 
        \item \textit{WMMSE with greedy scheduling}: The WMMSE algorithm has been well studied in the literature as a coordinated beamforming scheme; adaptive power allocation is implicitly included in the beamformer design. However, as noted in~\cite{shen2018fractional2}, WMMSE is intended for use as a beamforming algorithm; the question of which users to schedule remains to be answered. Accordingly, we adopt the greedy proportionally fair scheduling scheme introduced in~\cite{yu2013multicell}. 
        
        In this scheduling scheme, we first initialize the algorithm with a random set of $M$ users. The beamforming weights are held fixed and we sequentially determine which user will maximize the weighted rate on each beamforming weight from the BS, i.e., the ${k}^{\mathrm{th}}$ user associated with the ${b}^{\mathrm{th}}$ BS is scheduled on the ${j}^{\mathrm{th}}$ beam on the ${f}^{\mathrm{th}}$ frequency band if 
        \[
        {k}\mathrm{{{}={}}}\arg \hspace{-0.50em}\mathop{\max}\limits_{{i}\mathrm{{{}={}}}{1}{\mathrm{,...}}{K}}\hspace{0.33em}{\hat{r}}_{ib\mathrm{,}j\mathrm{,}f}
        \] This approach to scheduling is distinct from solving the linear sum assignment problem in the proposed algorithm, as the users are selected greedily for each beam rather than jointly across the set of all given beams.

        Importantly, this algorithm is not necessarily monotonically non-decreasing.
        
        \item \textit{Multicell WMMSE}: In multicell WMMSE~\cite{shi2011iteratively}, each BS initializes the algorithm by simultaneously scheduling \textit{all the users} in the network. The algorithm iterates on the beamformer design for all these users; eventually the beamforming weights for the majority of users will converge to zero, and these users are then implicitly not scheduled by the BS. This multicell WMMSE scheme is the state-of-the-art in the literature and has the same convergence properties as our algorithm. However, this comes at a cost: in order to determine the beamforming weights for each user, WMMSE performs a matrix inversion and bisection search. With the multicell WMMSE scheme, since all the users in the network are scheduled, the number of matrix inversions becomes extremely high. This is especially inefficient as the number of users ultimately assigned beamforming weights with nonzero power is very small. As we will show in the analysis of the results, our proposed algorithm is capable of outperforming multicell WMMSE, while simultaneously providing significant savings in computational complexity.
\end{enumerate}
{We consider a network partitioned into identical hexagonal cells, with BSs located at the center of each cell. The users are distributed randomly with uniform density over the entire network area. Furthermore, we also assume that the number of users associated with each BS significantly exceeds the number of transmit antennas available at the BS (i.e., ${K}_{b}\mathrm{\gg}{M}$ for all $b$).} 

To compare the performance of the aforementioned resource allocation schemes, we simulate a 7-cell network with wraparound. To ensure a fair comparison, all iterative optimization schemes were run for 15 iterations. The rest of the simulation parameters are as listed in Table \ref{tab1}.
\begin{table}[b!]
        \centering
        \caption{The numerical values of parameters used in the system model.} \label{tab1}
        \begin{tabular}{ |c|c| } 
                \hline
                Total bandwidth & \textit{W} = 20 MHz \\  \hline
                BS maximum average transmit power per frequency band & ${P}_{T}$ = 43 dBm\\ \hline
                Noise figure & ${N}_{f}$ = 9 dB \\ \hline
                Path-loss exponent &  $\mathit{\alpha}\hspace{0.33em}\mathrm{{{}={}}}\hspace{0.33em}{3}{\mathrm{.}}{\mathrm{76}}$\\ \hline
                Reference distance & 0.3920 m \\ \hline
        \end{tabular}
\end{table}

{We begin by comparing the performance of the proposed algorithm against the benchmark schemes listed, as well as against standard interior-point and sequential quadratic programming algorithms utilized in the literature \cite{nocedal2006numerical,fletcher1987practical}. The latter were implemented using standard available optimization software; all users were scheduled in a similar fashion to the multicell WMMSE algorithm. Figure \ref{WSR_convergence_M2} shows the convergence of network sum-rate for \textit{M}=2 transmit antennas and $K_b$=5 users per cell.
\begin{figure}[!t] 
        \begin{center} 
                \setlength\belowcaptionskip{-3.7\baselineskip}
                \includegraphics[trim={02cm 06cm 2cm 7cm},clip, height=0.55\textwidth]{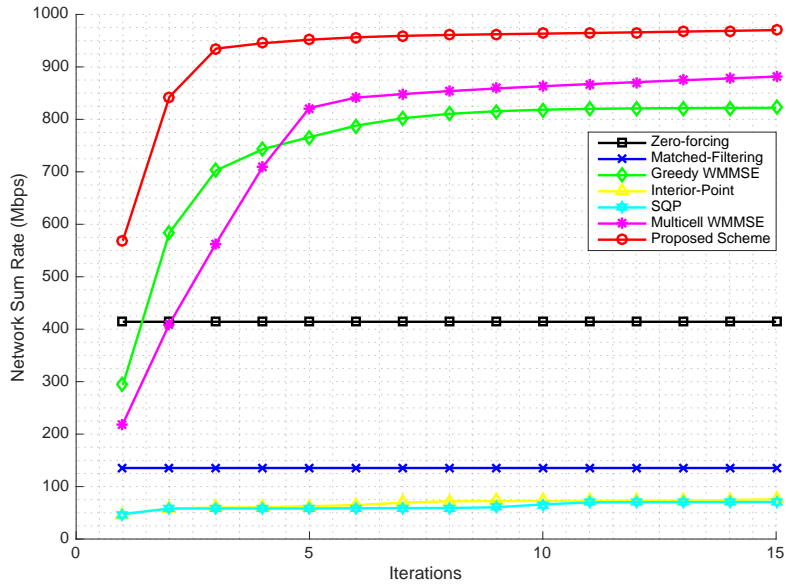}
                \caption{{Convergence of Network Sum Rate for Different Resource Allocation Schemes for $M$=2, $K_b$=5.}}
                \centering
                \centering
                \label{WSR_convergence_M2}
        \end{center}
\end{figure} 
As we can observe, the proposed algorithm achieves a higher network sum-rate, converging smoothly in a monotonically non-decreasing fashion; this is as expected from Theorem 2. At the same time, it is also clear that the uncoordinated resource allocation strategies perform substantially worse than the coordinated strategies, with matched filtering being the worst beamforming strategy. In addition, we observe that there is a gap in performance between greedy and multicell WMMSE. This is readily understood since the greedy scheduling approach is \textit{not} guaranteed to increase the network WSR after the scheduling reassignment. The multicell WMMSE algorithm is guaranteed to converge in a monotonically non-decreasing fashion, and as such provides good performance. Nevertheless, it is outperformed by the proposed scheme, which converges to the highest network weighted sum rate of all resource allocation schemes. In contrast, the sequential quadratic programming and interior-point algorithms show improving objective values as the number of iterations increase; however, they provide the worst performance. Due to the highly non-convex nature of the WSR-max optimization problem, these methods have demonstrated inferior performance in prior works in this area \cite{chitti2013joint,shen2018fractional2}; hence these results are not unexpected.}
\begin{figure}[!t] 
        \begin{center} 
                \setlength\belowcaptionskip{-3.7\baselineskip}
                \includegraphics[trim={5cm 3cm 5cm 4cm},clip, height=0.55\textwidth]{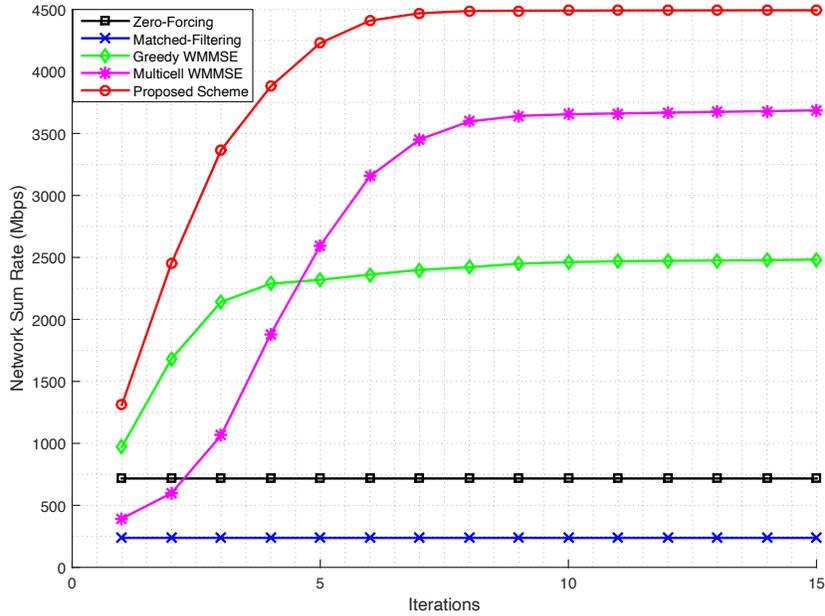}
                \centering
                \caption{Convergence of Network Sum Rate for Different Resource Allocation Schemes.}
                \centering
                \label{WSR_convergence}
        \end{center}
\end{figure} 

\emph{Convergence:} Figure \ref{WSR_convergence} shows the sum-rate convergence of the various resource allocation schemes for a single time slot with identical channel sets but with a much larger network size of $M=8$ and $K_b=80$. We observe similar trends to those in Figure \ref{WSR_convergence_M2}, with two notable differences. First, the greedy WMMSE algorithm is substantially outperformed by the multicell WMMSE algorithm, and the performance gap between the proposed scheme and the benchmarks grows larger as well. Secondly, due to the large number of optimization variables involved, the SQP and interior-point algorithms failed to converge for this setting. In particular, for the SQP approach, taking a direct step involves computing the Hessian of the optimization variables, which is extremely computationally demanding for a problem of this given size \cite{nocedal2006numerical}.

It is important to emphasize that none of the schemes result in the globally optimal solution and, if the number of potential users to be scheduled is very large, the WMMSE algorithm can get stuck in a poor solution and often takes longer to converge. In contrast, the proposed scheme restricts the BS to serve at most $M$ users, thereby narrowing the pool of potential users and ensuring faster convergence to a higher-quality local optimum.
\begin{figure}[!t] 
        \begin{center} 
                \setlength\belowcaptionskip{-3.7\baselineskip}
                \includegraphics[trim={0cm 5cm 0cm 5cm},clip, height=0.65\textwidth]{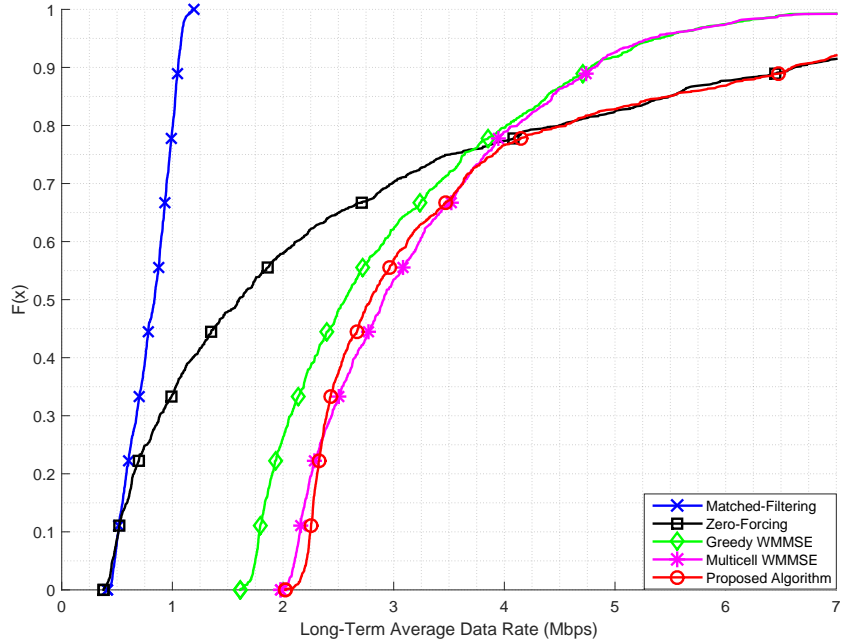}
                \centering
                \caption{CDFs of average data rates achieved with different resource allocation schemes for a fixed number of iterations.}
                \centering
                \centering
                \label{LS_MIMO_CDF}
        \end{center}
\end{figure}

\emph{PF Rates:} In Figure 3, we present the cumulative distribution functions (CDFs) of the long-term user average data rates achieved with the different resource allocation schemes (with $N_{iterations}=15$, $M=8$ and $K_b=80$). In order to compare these, we consider two metrics: the sum of the logarithm of the long-term average data rates (in megabits per second) and the ${\mathrm{10}}^{\mathrm{th}}$ percentile user rates, which are logged in Table \ref{LS_MIMO_sumlog_edge_table}. We choose to maximize the WSR in each time slot when the weights are chosen according to the proportionally fair metric described earlier (with the forgetting factor $\alpha$ chosen as 0.05), as this leads to maximization of the sum of the logarithm of the average data rates achieved by the users. Thus, comparing the average sum-log utility of the different resource allocation schemes allows us to directly compare them in terms of our original objective. We note that comparing the absolute difference (rather than the relative gain) in the sum-log-utility of different schemes illustrates the improvements made to the average rates achieved by the users. 

Comparing the ${\mathrm{10}}^{\mathrm{th}}$ percentile user rates allows us to compare the quality of service for the cell-edge users for the different resource allocation schemes. It is important to note, however, that the algorithms do not optimize the cell-edge rate; comparing edge user rates merely allows us to understand the quality of service that these different resource allocation schemes provide to the lower-percentile users in the network. 

\indent As we can observe from Figure 3 and Table II, the uncoordinated resource allocation schemes have the worst performance in terms of both the average sum-log-utility and edge user rates. This is unsurprising, since the benefits of coordinated resource allocation schemes over uncoordinated schemes are well-known. We note that employing zero-forcing results in significantly higher average sum log utility than matched filtering; this is also to be expected, since the former scheme eliminates intracell interference. 
All three coordinated resource allocation schemes achieve significantly higher performance than the uncoordinated schemes. However, of the two WMMSE resource allocation schemes, greedy scheduling has the worst performance; this is because greedy scheduling is sub-optimal and the associated algorithm is not guaranteed to be monotonic.

Utilizing multicell WMMSE results in a further significant gain to both the average sum-log-utility and the edge user rates. The proposed scheme performs even better than the multicell WMMSE scheme with considerably higher sum-log-utility and slightly better edge rates. The majority of the performance gain comes at higher percentiles, where the proposed approach achieves much better data rates than multicell WMMSE. Indeed, compared to the uncoordinated resource allocation schemes, there is a fourfold increase in the $10^{th}$ percentiles, with an increase of almost 30\% compared to the greedy WMMSE scheme. The sum-log-utility of the proposed scheme is also considerably higher than that achieved by the greedy WMMSE scheme. 
\begin{table}[!t]
        \centering
        \caption{Network sum-log utility and edge user rates for different resource allocation schemes with fixed number of iterations.}
        \label{LS_MIMO_sumlog_edge_table}
        \begin{tabular}{|
                        >{\columncolor[HTML]{EFEFEF}}l |
                        >{\columncolor[HTML]{FFFFFF}}c |
                        >{\columncolor[HTML]{FFFFFF}}c |}
                \hline
                \multicolumn{1}{|c|}{\cellcolor[HTML]{EFEFEF}\textbf{\begin{tabular}[c]{@{}c@{}}Resource Allocation \\ Strategy\end{tabular}}} & \cellcolor[HTML]{EFEFEF}\textbf{\begin{tabular}[c]{@{}c@{}}Average Network \\ Sum-Log Utility\end{tabular}} & \cellcolor[HTML]{EFEFEF}\textbf{\begin{tabular}[c]{@{}c@{}}Edge User Data Rate\\ (Mbps)\end{tabular}} \\ \hline
                Matched Filtering & -203 & 0.51 \\ \hline
                Zero-Forcing & 446 & 0.51 \\ \hline
                Greedy WMMSE & 821 & 1.78 \\ \hline
                Multicell WMMSE & 908 & 2.15 \\ \hline
                Proposed Algorithm & 952 & 2.25 \\ \hline
        \end{tabular}
\end{table}

A key point related to Table~\ref{LS_MIMO_sumlog_edge_table} is that we initialize the proposed algorithm by scheduling the set of users that achieves the highest interference-free weighted sum rate with an equal power allocation. For a worst-case initialization (i.e., if we start with the set of users that achieves the lowest interference-free weighted sum rate), the sum-log utility function is 909 for the proposed algorithm, still higher than that achieved by the multicell WMMSE algorithm. We emphasize that there is no known optimal initialization for the coordinated resource allocation schemes.  

\emph{Sum Rate:} To change the optimization objective function, we compare the performance of the aforementioned resource allocation schemes in terms of network sum-rate when the BS maximum transmit power, $P_T$ is varied from 20 to 70 dBm for $M=8$ and $K_b=80$. When maximizing the sum-rate, all users' weights are set to unity and this assignment does  not change across time-slots. The corresponding results are shown in Figure \ref{sumrates}.
\begin{figure}[!t] 
        \begin{center} 
                \setlength\belowcaptionskip{-3.7\baselineskip}
                \includegraphics[trim={01cm 06cm 1cm 7cm},clip, height=0.55\textwidth]{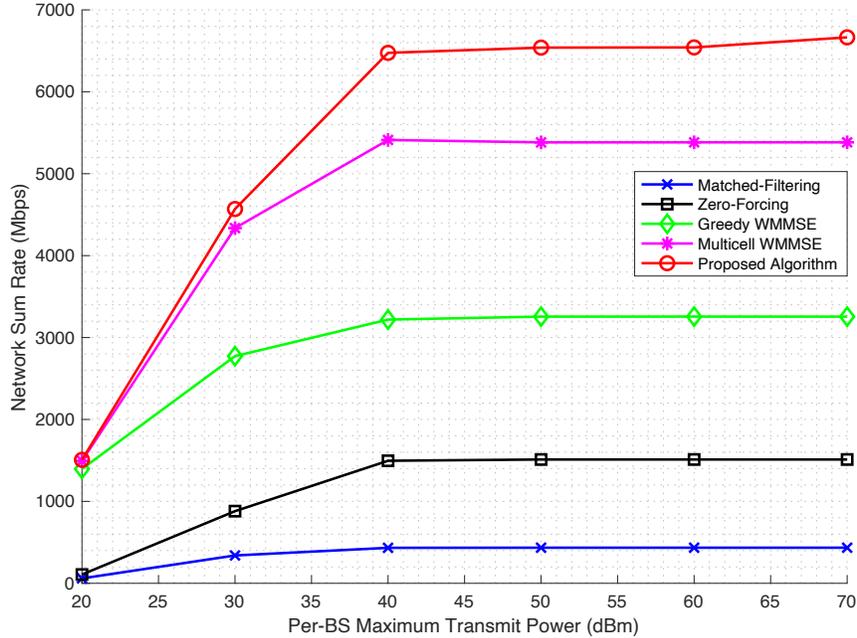}
                \centering
                \caption{Network sum-rate as a function of per-BS transmit power for different resource allocation schemes.}
                \centering
                \centering
                \label{sumrates}
        \end{center}
\end{figure}

As the figure clearly shows, the proposed algorithm substantially outperforms the competing WMMSE approaches. In particular, there is a gap of approximately 22\% in network sum-rate across transmit powers above 40dBm between the proposed approach and the benchmark multicell WMMSE algorithm. Both algorithms substantially outperform the greedy WMMSE algorithm, delivering sum-rates that are more than twice as high for large BS transmit powers. As expected, the performance the matched filtering and zero-forcing schemes lags behind those of the coordinated schemes.

This result in particular demonstrates the performance advantage our optimal scheduling approach based on the Hungarian algorithm delivers over greedy scheduling. One interesting phenomenon to note is that the network-sum rate of the proposed algorithm strictly increases as a function of BS transmit power; however, this is not always the case for the multicell WMMSE algorithm. This result highlights the tendency of the multicell WMMSE algorithm to get stuck in low-quality local optima. 

\emph{Optimization Across Frequency Bands:} Finally, we compare the performance of the joint power allocation approach versus the per-band power allocation method. In the former setting, we assign a total power of $FP_T$ to be distributed over the $F$ frequency bands at each BS. This means that the BS is free to use as much or as little power in each of the frequency bands, provided that the total power consumed across all frequency bands is less than $FP_T$. In contrast, with the per-band power allocation strategy, each BS can only utilize a maximum power of $P_T$ per individual frequency band. As we can observe in Fig. \ref{multifr}; there is no significant performance benefit in terms of either edge rates or overall utility to choosing the joint power allocation strategy over the per-band strategy. Furthermore, the number of iterations to calculate the beamforming weights is more than the decoupled setting, since the bisection search step is now being performed across all frequency bands.

\begin{figure}[!t] 
        \begin{center} 
                \setlength\belowcaptionskip{-3.7\baselineskip}
                \includegraphics[trim={0cm 5cm 0cm 5cm},clip, height=0.65\textwidth]{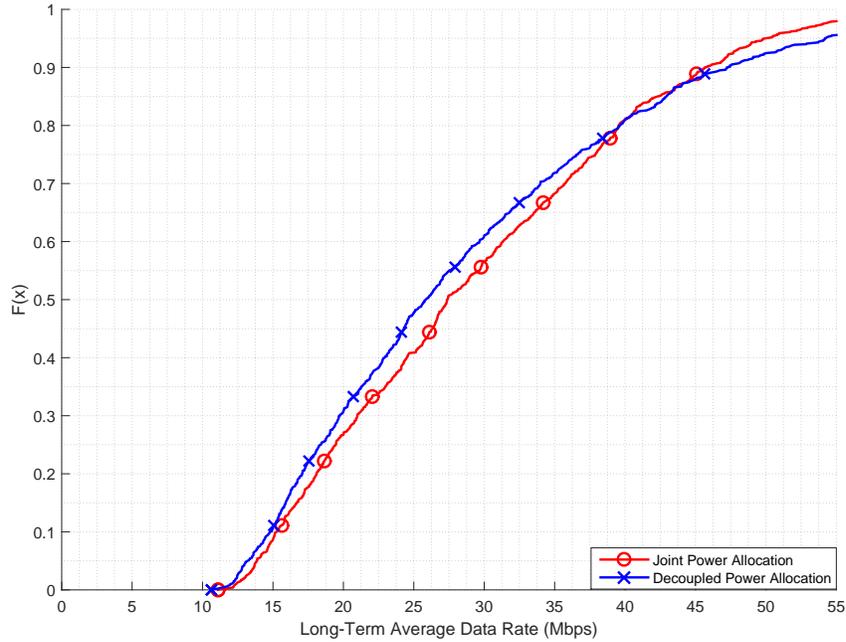}
                \centering
                \caption{CDFs of average data rates achieved with joint and decoupled power allocation schemes.}
                \centering
                \centering
                \label{multifr}
        \end{center}
\end{figure}

\subsection{Complexity Analysis}
In comparing the performance of these various resource allocation strategies, a critical point is the computational complexity involved. From~\cite{shi2011iteratively}, the computational complexity of the beamforming step in WMMSE can be derived as ${\mathcal{O}}{\mathrm{(}}{\mathit{\kappa}}^{2}{M}^{2}\mathrm{{+}}{\mathit{\kappa}}{M}^{3}{\mathrm{)}}$, where $\mathit{\kappa}$ represents the total number of users scheduled in the network. For our setting, we have  $\left|{\mathcal{B}}\right|$ cooperating cells. For simplicity of analysis, we assume that each cell has $K$ users associated with it; thus we have $\mathit{\kappa}\mathrm{{{}={}}}{K}\left|{\mathcal{B}}\right|$ for the multicell WMMSE scheme and $\mathit{\kappa}\mathrm{{{}={}}}{M}\left|{\mathcal{B}}\right|$ for the greedy WMMSE scheme and proposed.

Accordingly, the computational complexity of the WMMSE algorithm with greedy scheduling can be found as ${\mathcal{O}}{\mathrm{(}}{M}^{4}{\left|{\mathcal{B}}\right|}^{2}\mathrm{{+}}{M}^{2}{K}{\left|{\mathcal{B}}\right|}^{2}{\mathrm{)}}$ whereas the computational complexity of the multicell WMMSE algorithm is $
{\mathcal{O}}{\mathrm{(}}{M}^{3}{K}\left|{\mathcal{B}}\right|\mathrm{{+}}{M}^{2}{K}^{2}{\left|{\mathcal{B}}\right|}^{2}{\mathrm{)}}
$. The proposed algorithm has the same computational complexity as the greedy WMMSE scheme (since we schedule at most $M$ users in a single time slot). It follows that the computational complexity of the fractional programming strategy is at most given by ${\mathcal{O}}{\mathrm{(}}{M}^{4}{\left|{\mathcal{B}}\right|}^{2}\mathrm{{+}}{M}^{2}{K}{\left|{\mathcal{B}}\right|}^{2}{\mathrm{)}}$. A comparison of the per-iteration computational complexity of the various resource allocation schemes discussed is provided in Table \ref{complexity_analysis}.
\begin{table}[!t]
        \centering
        \caption{Per-iteration computational complexity of different resource allocation schemes.}
        \label{complexity_analysis}
        \begin{tabular}{|
                        >{\columncolor[HTML]{EFEFEF}}l |
                        >{\columncolor[HTML]{FFFFFF}}c |}
                \hline
                \multicolumn{1}{|c|}{\cellcolor[HTML]{EFEFEF}\textbf{\begin{tabular}[c]{@{}c@{}}Resource Allocation \\ Strategy\end{tabular}}} & \cellcolor[HTML]{EFEFEF}\textbf{Complexity Per Iteration} \\ \hline
                Matched Filtering\footnotemark & ${\mathcal{O}}{\mathrm{(}}{M}^{2}\left|{\mathcal{B}}\right|{\mathrm{)}}$ \\ \hline
                Zero-Forcing\footnotemark[\value{footnote}] &  ${\mathcal{O}}{\mathrm{(}}{M}^{4}\left|{\mathcal{B}}\right|{\mathrm{)}}$\\ \hline
                Greedy WMMSE & ${\mathcal{O}}{\mathrm{(}}{M}^{4}{\left|{\mathcal{B}}\right|}^{}\mathrm{{+}}{M}^{2}{K}{\left|{\mathcal{B}}\right|}^{2}{\mathrm{)}}$ \\ \hline
                Multicell WMMSE & $\mathcal{O}\mathrm{(}{M}^{3}{K}{\left|{\mathcal{B}}\right|}\mathrm{}\mathrm{{+}}{M}^{2}{K}^{2}{\left|{\mathcal{B}}\right|}^{2}{\mathrm{)}}$\\ \hline
                Proposed Algorithm & ${\mathcal{O}}{\mathrm{(}}{M}^{4}{\left|{\mathcal{B}}\right|}^{}\mathrm{{+}}{M}^{2}{K}{\left|{\mathcal{B}}\right|}^{2}{\mathrm{)}}$\\ \hline
        \end{tabular}
\end{table}
\footnotetext{Note that these uncoordinated schemes require only a single iteration to determine the network resource allocation strategy.}

To understand these results, we revisit some of the assumptions made earlier in the system model. Since we deal with a large-scale MIMO system, the number of users in each cell ($K$) is assumed to be significantly larger than the number of antennas at the BS. {Critically, as $K$ becomes asymptotically large, for multicell WMMSE, the ${M}^{2}{K}^{2}{\left|{\mathcal{B}}\right|}^{2}$ term will dominate the complexity expression. On the other hand, for the proposed algorithm, the only term dependent upon the number of users per cell is ${M}^{2}{K}{\left|{\mathcal{B}}\right|}^{2}$; thus, if $M$ is fixed and $K\gg{M}$ as per the assumption in our system model earlier, then the $M^{4}|\mathcal{B}|^{2}$ term becomes negligible in comparison. It follows that in this case, the complexity of the multicell WMMSE algorithm will be roughly $K/M$ times higher than the proposed algorithm. Our algorithm achieves this noteworthy reduction in complexity while exceeding the performance of the multicell WMMSE algorithm.

It is worth noting that even with large computation resources in a CRAN, a fully centralized globally optimal solution is infeasible since the problem at hand is $\textrm{NP}$-hard. Furthermore, even in this case, the reduced computational complexity of our algorithm, compared to say the multicell WMMSE approach, is important for implementation with a large number of users and BS antennas. We emphasize that, furthermore, this gain in computational complexity is accompanied by improved performance. Compared to the generic solvers we have considered, our proposed algorithm is more convenient from an implementation perspective since the updates for the optimization variables are expressed in closed-form. This is also in contrast to globally optimal strategies such as outer polyblock approximation, in which the updates are not expressible in closed-form \cite{utschick2012monotonic}. We also note that the SQP and interior-point algorithms do not have closed-form updates; calculating Hessians for the latter approach in particular is computationally taxing for large network sizes \cite{nocedal2006numerical}.}

Prior to proceeding further, we consider the reason for this behavior in greater detail. Recall that in the multicell WMMSE scheme, we schedule all users simultaneously and let the beamformer design iterate. Thus, this is equivalent to considering the original optimization problem but with \textit{no scheduling constraints}. As the algorithm converges, most users are assigned beamformers with zero power; the number of users ultimately assigned nonzero power is very close to $M$. Nonetheless, the beamforming weights still have to be calculated for the users who will ultimately be dropped since they are not known \textit{a priori}. This requires computationally costly matrix inversions, thus leading to a higher overall complexity. With the proposed algorithm, since we schedule the best $M$ users in a single time slot, the number of matrix inversions required is identical to the greedy WMMSE strategy. 

{A second consideration is that the set of users scheduled in each iteration of the algorithm has the potential to change. Unlike the greedy WMMSE strategy, however, the proposed scheme ensures that the network WSR increases after each scheduling step as we find the \emph{best network-wide scheduling pattern} for the fixed set of beamforming weight vectors. Hence, we conclude that the proposed strategy of intelligently scheduling the smaller set of users (which is close to the number of users implicitly scheduled by the multicell WMMSE scheme) nets an improvement in terms of computational complexity while providing superior performance. Also, it is worth pointing out that scheduling all users, as the multicell WMMSE algorithm does, requires a much greater overhead in terms of communication between the coordinated BSs in the network.}

\begin{table}[]
        \centering
        \caption{{Average execution time of different resource allocation schemes.}}
        \label{execution_time}
        \begin{tabular}{|
                        >{\columncolor[HTML]{EFEFEF}}l |c|c|}
                \hline
                \textbf{Resource Allocation Scheme} &
                \multicolumn{1}{l|}{\cellcolor[HTML]{EFEFEF}\text{$M=2$, $K_b=5$}} &
                \multicolumn{1}{l|}{\cellcolor[HTML]{EFEFEF}\text{$M=8$, $K_b=80$}} \\ \hline
                Matched Filtering  & 0.01 & 1.1  \\ \hline
                Zero-Forcing       & 0.05 & 1.2  \\ \hline
                Greedy WMMSE       & 1.3  & 27.2 \\ \hline
                Multicell WMMSE    & 1.3  & 48.8 \\ \hline
                Proposed Algorithm & 1.3  & 19.5 \\ \hline
                Interior-Point     & 40.6 & N/A  \\ \hline
                SQP                & 53.6 & N/A  \\ \hline
        \end{tabular}
\end{table}
\footnotetext{These execution times are obtained using a desktop computer with a 3.6 GHz Intel{\textregistered} Core{\texttrademark} i7-4790 CPU and 24 GB of RAM.}

Finally, we consider the actual execution time of the various resource allocation schemes. Although the complexity analysis provides a formal characterization of how the running time of each resource allocation scheme scales as the network parameters change, it is nonetheless useful for us to compare the average execution time of each scheme. To ensure a fair comparison, we measure the time taken from initialization until the per-iteration increase in the network weighted sum rate is less than 10\% for the given time slot. For the simulation parameters detailed in Table \ref{tab1}, the average execution times on the desktop computer used to generate these results are logged in Table \ref{execution_time}. As we can see, the uncoordinated schemes require a much lower execution time on average than the coordinated schemes; however, this comes at the expense of compromised performance as discussed earlier. Both the proposed algorithm and greedy WMMSE approach perform similarly in terms of average execution time. The multicell WMMSE scheme has the highest average execution time among coordinated schemes; as discussed earlier, this is due to the fact that all users in the network are scheduled simultaneously, so the number of matrix inversions needed is very large compared to both greedy WMMSE and the proposed algorithm.{ Finally, the interior-point and SQP algorithms have extremely long execution times for the $M=2$ and $K_b=5$ setting and do not converge within a reasonable time for the $M=8$ and $K_b=80$ setting.}

\section{Conclusions}
        In this paper, we developed a coordinated resource allocation scheme for the downlink of multiuser MIMO networks with multiple orthogonal frequency bands. The proposed scheme outperforms uncoordinated schemes like zero-forcing and matched filtering, as well as the coordinated greedy and state-of-the-art multicell WMMSE schemes in terms of the average sum-log-utility function and network sum-rate. Furthermore, the proposed scheme offers significant computational complexity savings over the state-of-the-art multicell WMMSE scheme and also has a much lower average execution time. By intelligently scheduling the best subset of $M$ users for a fixed given set of beamforming weights, the proposed approach is able to reduce the computational complexity as well as providing a higher weighted sum-rate in a single time slot and higher long-term average sum-log-utility. Thus, we conclude that the proposed approach offers an effective high-performance and low-complexity solution to the nonconvex NP-hard weighted sum-rate maximization problem.
        {
\begin{spacing}{1.0}
\bibliographystyle{IEEEtran}
\bibliography{IEEEabrv,biblio}  
\end{spacing}}
\end{document}